\documentclass[11pt]{article}

\usepackage[dvips]{graphicx, color}

\usepackage{amsmath} 
\usepackage{amssymb}
\usepackage{amsthm}
\usepackage{amsxtra}
\usepackage{amscd}
\usepackage{bbm}
\usepackage{mathrsfs}
\usepackage{bm}

\def\bfseries{\fontseries \bfdefault \selectfont \boldmath}

\def\bfdefault{b}

\renewcommand{\cal}{\mathcal}

\newcommand{\ol}[1]{\overline{#1} \!\,} 

\newcommand{\ee}{\mathrm{e}} 
\newcommand{\ii}{\mathrm{i}} 
\newcommand{\dd}{\mathrm{d}}

\newcommand{\deq}{\mathrel{\mathop:}=}
\newcommand{\eqd}{=\mathrel{\mathop:}}
\newcommand{\umat}{\mathbbmss{1}} 
\renewcommand{\epsilon}{\varepsilon}
\renewcommand{\leq}{\leqslant}
\renewcommand{\geq}{\geqslant}

\newcommand{\R}{\mathbb{R}}
\newcommand{\C}{\mathbb{C}}
\newcommand{\N}{\mathbb{N}}
\newcommand{\Z}{\mathbb{Z}}

\newcommand{\p}[1]{({#1})}
\newcommand{\pb}[1]{\bigl({#1}\bigr)}
\newcommand{\pB}[1]{\Bigl({#1}\Bigr)}
\newcommand{\pbb}[1]{\biggl({#1}\biggr)}

\newcommand{\qb}[1]{\bigl[{#1}\bigr]}
\newcommand{\qB}[1]{\Bigl[{#1}\Bigr]}
\newcommand{\qbb}[1]{\biggl[{#1}\biggr]}

\newcommand{\abs}[1]{\lvert #1 \rvert}
\newcommand{\absb}[1]{\big\lvert #1 \big\rvert}
\newcommand{\absB}[1]{\Big\lvert #1 \Big\rvert}
\newcommand{\absbb}[1]{\bigg\lvert #1 \bigg\rvert}

\newcommand{\norm}[1]{\lVert #1 \rVert}
\newcommand{\normb}[1]{\big\lVert #1 \big\rVert}
\newcommand{\normB}[1]{\Big\lVert #1 \Big\rVert}

\newcommand{\scalar}[2]{\langle{#1} \mspace{2mu}, {#2}\rangle}
\newcommand{\scalarb}[2]{\big\langle{#1} \mspace{2mu}, {#2}\big\rangle}
\newcommand{\scalarB}[2]{\Big\langle{#1} \,\mspace{2mu},\, {#2}\Big\rangle}
\newcommand{\scalarbb}[2]{\bigg\langle{#1} \,\mspace{2mu},\, {#2}\bigg\rangle}

\newcommand{\com}[2]{[{#1} \mspace{2mu}, {#2}]}
\newcommand{\comb}[2]{\big[{#1} \mspace{2mu}, {#2}\big]}

\newcommand{\combb}[2]{\bigg[{#1} \,\mspace{2mu},\, {#2}\bigg]}

\newcommand{\bra}[1]{\langle #1 |}

\newcommand{\ket}[1]{| #1 \rangle}

\DeclareMathOperator{\tr}{Tr}

\DeclareMathOperator{\re}{Re}

\theoremstyle{plain} 
\newtheorem{theorem}{Theorem}[section]
\newtheorem*{theorem*}{Theorem}
\newtheorem{lemma}[theorem]{Lemma}
\newtheorem*{lemma*}{Lemma}
\newtheorem{corollary}[theorem]{Corollary}
\newtheorem*{corollary*}{Corollary}

\newtheorem*{proposition*}{Proposition}
\theoremstyle{definition} 

\newtheorem*{definition*}{Definition}

\newtheorem*{example*}{Example}
\newtheorem{remark}[theorem]{Remark}

\newtheorem*{remark*}{Remark}
\newtheorem*{remarks*}{Remarks}

\usepackage{cite}

\usepackage[dvips]{graphicx}
\usepackage[nottoc,notlof,notlot]{tocbibind}

\usepackage{color}
\usepackage{psfrag}

\numberwithin{equation}{section}
\newcommand{\step}[1]{\vspace{1ex} \noindent {\itshape #1}}

\begin{document}
\title{Mean-Field Dynamics: Singular Potentials and Rate of Convergence}
\author{Antti Knowles and Peter Pickl}
\maketitle

\begin{abstract}
We consider the time evolution of a system of $N$ identical bosons whose interaction potential 
is rescaled by $N^{-1}$. We choose the initial wave function to describe a condensate in which 
all particles are in the same one-particle state. It is well known that in the mean-field limit 
$N \to \infty$ the quantum $N$-body dynamics is governed by the nonlinear Hartree equation. 
Using a nonperturbative method, we extend previous results on the mean-field limit in two 
directions. First, we allow a large class of singular interaction potentials as well as strong, 
possibly time-dependent external potentials. Second, we derive bounds on the rate of 
convergence of the quantum $N$-body dynamics to the Hartree dynamics.
\end{abstract}

\section{Introduction}
We consider a system of $N$ identical bosons in $d$ dimensions, described by a wave function 
$\Psi_N \in \cal H^{(N)}$. Here
\begin{equation*}
\cal H^{(N)} \;\deq\; L^2_+(\R^{Nd}, \dd x_1 \cdots \dd x_N)
\end{equation*}
is the subspace of $L^2(\R^{Nd}, \dd x_1 \cdots \dd x_N)$ consisting of wave functions 
$\Psi_N(x_1, \dots, x_N)$ that are symmetric under permutation of their arguments $x_1, \dots, 
x_N \in \R^d$. The Hamiltonian is given by \begin{equation} \label{mean-field Hamiltonian in 
introduction}
H_N \;=\; \sum_{i = 1}^N h_i + \frac{1}{N} \sum_{1 \leq i < j \leq N} w(x_i - x_j)\,,
\end{equation}
where $h_i$ denotes a one-particle Hamiltonian $h$ (to be specified later) acting on the 
coordinate $x_i$, and $w$ is an interaction potential. Note the mean-field scaling $1/N$ in 
front of the interaction potential, which ensures that the free and interacting parts of $H_N$ 
are of the same order. 

The time evolution of $\Psi_N$ is governed by the $N$-body Schr\"odinger equation
\begin{equation} \label{Schrodinger equation for Psi}
\ii \partial_t \Psi_N(t) \;=\; H_N \Psi_N(t) \,, \qquad \Psi_N(0) \;=\; \Psi_{N,0}\,.
\end{equation}
For definiteness, let us consider factorized initial data $\Psi_{N,0} = \varphi_0^{\otimes N}$ 
for some $\varphi_0 \in L^2(\R^d)$ satisfying the normalization condition 
$\norm{\varphi_0}_{L^2(\R^d)} = 1$. Clearly, because of the interaction between the particles, 
the factorization of the wave function is not preserved by the time evolution. However, it 
turns out that for large $N$ the interaction potential experienced by any single particle may 
be approximated by an effective mean-field potential, so that the wave function $\Psi_N(t)$ 
remains approximately factorized for all times. In other words we have that, in a sense to be 
made precise, $\Psi_N(t) \approx \varphi(t)^{\otimes N}$ for some appropriate $\varphi(t)$. A 
simple argument shows that in a product state $\varphi(t)^{\otimes N}$ the interaction 
potential experienced by a particle is approximately $w * \abs{\varphi(t)}^2$, where $*$ 
denotes convolution. This implies that $\varphi(t)$ is a solution of the nonlinear Hartree 
equation
\begin{equation} \label{hartree equation}
\ii \partial_t \varphi(t) \;=\; h \varphi(t) + \pb{w * \abs{\varphi(t)}^2} \varphi(t)\,, \qquad 
\varphi(0) \;=\; \varphi_0\,.
\end{equation}

Let us be a little more precise about what one means with $\Psi_N \approx \varphi^{\otimes N}$ 
(we omit the irrelevant time argument). One does not expect the $L^2$-distance $\normb{\Psi_N - 
\varphi^{\otimes N}}_{L^2(\R^{Nd})}$ to become small as $N \to \infty$. A more useful, weaker, 
indicator of convergence should depend only on a finite, fixed\footnote{In fact, as shown in 
Corollary \ref{corollary of main theorem}, $k$ may be taken to grow like $o(N)$.} number, $k$, 
of particles. To this end we define the reduced $k$-particle density matrix
\begin{equation*}
\gamma_N^{(k)} \;\deq\; \tr_{k+1, \dots, N} \ket{\Psi_N} \bra{\Psi_N}\,,
\end{equation*}
where $\tr_{k+1, \dots, N}$ denotes the partial trace over the coordinates $x_{k+1}, \dots, 
x_N$, and $\ket{\Psi_N} \bra{\Psi_N}$ denotes (in accordance with the usual Dirac notation) the 
orthogonal projector onto $\Psi_N$. In other words, $\gamma_N^{(k)}$ is the positive trace 
class operator on $L^2_+(\R^{kd}, \dd x_1 \cdots \dd x_k)$ with operator kernel
\begin{equation*}
\gamma_N^{(k)}(x_1, \dots, x_k; y_1, \dots, y_k) \;=\; \int \dd x_{k+1} \cdots \dd x_N \; 
\Psi_N(x_1, \dots, x_N) \ol{\Psi_N(y_1, \dots, y_k, x_{k+1}, \dots, x_N)}\,.
\end{equation*}
The reduced $k$-particle density matrix $\gamma^{(k)}_N$ embodies all the information contained 
in the full $N$-particle wave function that pertains to at most $k$ particles.
There are two commonly used indicators of the closeness $\gamma_N^{(k)} \approx (\ket{\varphi} 
\bra{\varphi})^{\otimes k}$: the projection
\begin{equation*}
E^{(k)}_N \;\deq\; 1 - \scalarb{\varphi^{\otimes k}}{\gamma_N^{(k)} \varphi^{\otimes k}}
\end{equation*}
and the trace norm distance
\begin{equation}
R^{(k)}_N \;\deq\; \tr \absB{\gamma_N^{(k)} - (\ket{\varphi} \bra{\varphi})^{\otimes k}}\,.
\end{equation}
It is well known (see e.g.\ \cite{LiebSeiringer2002}) that all of these indicators are 
equivalent in the sense that the vanishing of either $R^{(k)}_N$ or $E^{(k)}_N$ for some $k$ in 
the limit $N \to \infty$ implies that $\lim_N R^{(k')}_N = \lim_N E^{(k')}_N = 0$ for all $k'$.  
However, the rate of convergence may differ from one indicator to another. Thus, when studying 
rates of convergence, they are not equivalent (see Section \ref{measures of convergence} below 
for a full discussion).

The study of the convergence of $\gamma^{(k)}_N(t)$ in the mean-field limit towards 
$(\ket{\varphi(t)} \bra{\varphi(t)})^{\otimes k}$ for all $t$ has a history going back almost 
thirty years. The first result is due to Spohn \cite{Spohn1980}, who showed that $\lim_N 
R_N^{(k)}(t) = 0$ for all $t$ provided that $w$ is bounded. His method is based on the BBGKY 
hierarchy,
\begin{multline} \label{BBGKY hierarchy}
\ii \partial_t \gamma_N^{(k)}(t) \;=\; \sum_{i = 1}^k \comb{h_i}{\gamma_N^{(k)}(t)} + 
\frac{1}{N} \sum_{1 \leq i < j \leq k} \comb{w(x_i - x_j)}{\gamma_N^{(k)}(t)}
\\
+ \frac{N-k}{N} \sum_{i = 1}^k \tr_{k+1} \comb{w(x_i - x_{k+1})}{\gamma_N^{(k+1)}(t)}\,,
\end{multline}
an equation of motion for the family $(\gamma_N^{(k)}(t))_{k \in \N}$ of reduced density 
matrices. It is a simple computation to check that the BBGKY hierarchy is equivalent to the 
Schr\"odinger equation \eqref{Schrodinger equation for Psi} for $\Psi_N(t)$. Using a 
perturbative expansion of the BBGKY hierarchy, Spohn showed that in the limit $N \to \infty$ 
the family $(\gamma_N^{(k)}(t))_{k \in \N}$ converges to a family $(\gamma_\infty^{(k)}(t))_{k 
\in \N}$ that satisfies the limiting BBGKY obtained by formally setting $N = \infty$ in 
\eqref{BBGKY hierarchy}. This limiting hierarchy is easily seen to be equivalent to the Hartree 
equation \eqref{hartree equation} via the identification $\gamma_\infty^{(k)}(t) = 
(\ket{\varphi(t)} \bra{\varphi(t)})^{\otimes k}$. We refer to \cite{ErdosYau2001} for a short 
discussion of some subsequent developments.

In the past few years considerable progress has been made in strengthening such results in 
mainly two directions. First, the convergence $\lim_N R_N^{(k)}(t) = 0$ for all $t$ has been 
proven for singular interaction potentials $w$. It is for instance of special physical interest 
to understand the case of a Coulomb potential, $w(x) = \lambda \abs{x}^{-1}$ where $\lambda \in 
\R$. The proofs for singular interaction potentials are considerably more involved than for 
bounded interaction potentials. The first result for the case $h = -\Delta$ and $w(x) = \lambda 
\abs{x}^{-1}$ is due to Erd\H{o}s and Yau \cite{ErdosYau2001}. Their proof uses the BBGKY 
hierarchy and a weak compactness argument. In \cite{ElgartSchlein2007}, Schlein and Elgart 
extended this result to the technically more demanding case of a semirelativistic kinetic 
energy, $h = \sqrt{\umat - \Delta}$ and $w(x) = \lambda \abs{x}^{-1}$. This is a critical case 
in the sense that the kinetic energy has the same scaling behaviour as the Coulomb potential 
energy, thus requiring quite refined estimates. A different approach, based on operator 
methods, was developed by Fr\"ohlich et al.\ in \cite{FrohlichKnowlesSchwarz2009}, where the 
authors treat the case $h = -\Delta$ and $w(x) = \lambda \abs{x}^{-1}$. Their proof relies on 
dispersive estimates and counting of Feynman graphs. Yet another approach was adopted by 
Rodnianski and Schlein in \cite{RodnianskiSchlein2007}. Using methods inspired by a 
semiclassical argument of Hepp \cite{Hepp1974} focusing on the dynamics of coherent states in 
Fock space, they show convergence to the mean-field limit in the case $h = -\Delta$ and $w(x) = 
\lambda \abs{x}^{-1}$.

The second area of recent progress in understanding the mean-field limit is deriving estimates 
on the rate of convergence to the mean-field limit. Methods based on expansions, as used in 
\cite{Spohn1980} and \cite{FrohlichKnowlesSchwarz2009}, give very weak bounds on the error 
$R_N^{(1)}(t)$, while weak compactness arguments, as used in \cite{ErdosYau2001} and 
\cite{ElgartSchlein2007}, yield no information on the rate of convergence. From a physical 
point of view, where $N$ is large but finite, it is of some interest to have tight error bounds 
in order to be able to address the question whether the mean-field approximation may be 
regarded as valid. The first reasonable estimates on the error were derived for the case $h = 
-\Delta$ and $w(x) = \lambda \abs{x}^{-1}$ by Rodnianski and Schlein in their work 
\cite{RodnianskiSchlein2007} mentioned above. In fact they derive an explicit estimate on the 
error of the form
\begin{equation*}
R_N^{(k)}(t) \;\leq\; \frac{C_1(k)}{\sqrt{N}} \, \ee^{C_2(k) t}
\end{equation*}
for some constants $C_1(k), C_2(k) > 0$. Using a novel approach inspired by Lieb-Robinson 
bounds, Erd\H{o}s and Schlein \cite{ErdosSchlein2008} further improved this estimate under the 
more restrictive assumption that $w$ is bounded and its Fourier transform integrable. Their 
result is
\begin{equation*}
R_N^{(k)}(t) \;\leq\;\frac{C_1}{N} \, \ee^{C_2 k} \ee^{C_3 t}\,,
\end{equation*}
for some constants $C_1, C_2, C_3 > 0$.

In the present article we adopt yet another approach based on a method of Pickl 
\cite{Pickl2009}. We strengthen and generalize many of the results listed above, by treating 
more singular interaction potentials as well as deriving estimates on the rate of convergence.  
Moreover, our approach allows for a large class of (possibly time-dependent) external 
potentials, which might for instance describe a trap confining the particles to a small volume.  
We also show that if the solution $\varphi(\cdot)$ of the Hartree equation satisfies a 
scattering condition, all of the error estimates are uniform in time.

The outline of the article is as follows. Section \ref{measures of convergence} is devoted to a 
short discussion of the indicators of convergence $E^{(k)}_N$ and $R^{(k)}_N$, in which we 
derive estimates relating them to each other.
In Section \ref{L2 potentials} we state and prove our first main result, which concerns the 
mean-field limit in the case of $L^2$-type singularities in $w$; see Theorem \ref{theorem for 
L^2 potentials} and Corollary \ref{corollary of main theorem}. In Section \ref{section: 
singular potentials} we state and prove our second main result, which allows for a larger class 
of singularities such as the nonrelativistic critical case $h = -\Delta$ and $w(x) = \lambda 
\abs{x}^{-2}$; see Theorem \ref{theorem for singular potentials}. For an outline of the methods 
underlying our proofs, see the beginnings of Sections \ref{L2 potentials} and \ref{section: 
singular potentials}.

\subsubsection*{Acknowledgements}
We would like to thank J.\ Fr\"ohlich and E.\ Lenzmann for helpful and stimulating discussions.  
We also gratefully acknowledge discussions with A.\ Michelangeli which led to Lemma \ref{lemma: 
relationship between Es}.

\subsubsection*{Notations}
Except in definitions, in statements of results and where confusion is possible, we refrain 
from indicating the explicit dependence of a quantity $a_N(t)$ on the time $t$ and the particle 
number $N$. When needed, we use the notations $a(t)$ and $a \vert_t$ interchangeably to denote 
the value of the quantity $a$ at time $t$. The symbol $C$ is reserved for a generic positive 
constant that may depend on some fixed parameters. We abbreviate $a \leq C b$ with $a \lesssim 
b$.
To simplify notation, we assume that $t \geq 0$.

We abbreviate $L^p(\R^d, \dd x) \equiv L^p$ and $\norm{\cdot}_{L^p} \equiv \norm{\cdot}_p$. We 
also set $\norm{\cdot}_{L^2(\R^{Nd})} = \norm{\cdot}$. For $s \in \R$ we use $H^s \equiv 
H^s(\R^d)$ to denote the Sobolev space with norm $\norm{f}_{H^s} = \normb{(1 + \abs{k}^2)^{s/2} 
\hat{f}}_2$, where $\hat{f}$ is the Fourier transform of $f$.

Integer indices on operators denote particle number: A $k$-particle operator $A$ (i.e.\ an 
operator on $\mathcal{H}^{(k)}$) acting on the coordinates $x_{i_1}, \dots, x_{i_k}$, where 
$i_1 < \dots < i_k$, is denoted by $A_{i_1 \dots i_k}$. Also, by a slight abuse of notation, we 
identify $k$-particle functions $f(x_1, \dots, x_k)$ with their associated multiplication 
operators on $\mathcal{H}^{(k)}$. The operator norm of the multiplication operator $f$ is equal 
to, and will always be denoted by, $\norm{f}_\infty$.

We use the symbol $\mathcal{Q}(\cdot)$ to denote the form domain of a semibounded operator. We 
denote the space of bounded linear maps from $X_1$ to $X_2$ by $\cal L(X_1; X_2)$, and 
abbreviate $\cal L(X) = \cal L(X;X)$. We abbreviate the operator norm of $\cal 
L\pb{L^2(\R^{Nd})}$ by $\norm{\cdot}$.
For two Banach spaces, $X_1$ and  $X_2$, contained in some larger space, we set
\begin{align*}
\norm{f}_{X_1 + X_2} &\;=\; \inf_{f = f_1 + f_2} \pb{\norm{f_1}_{X_1} + \norm{f_2}_{X_2}}\,,
\\
\norm{f}_{X_1 \cap X_2} &\;=\; \norm{f}_{X_1} + \norm{f}_{X_2}\,,
\end{align*}
and denote by $X_1 + X_2$ and $X_1 \cap X_2$ the corresponding Banach spaces.

\section{Indicators of convergence} \label{measures of convergence}
This section is devoted to a discussion, which might also be of independent interest, of 
quantitative relationships between the indicators $E^{(k)}_N$ and $R^{(k)}_N$. Throughout this 
section we suppress the irrelevant index $N$.

Take a $k$-particle density matrix $\gamma^{(k)} \in \mathcal{L}(\mathcal{H}^{(k)})$ and a 
one-particle condensate wave function $\varphi \in L^2$.
The following lemma gives the relationship between different elements of the sequence $E^{(1)}, 
E^{(2)}, \dots$, where, we recall,
\begin{equation} \label{definition of E}
E^{(k)} \;=\; 1 - \scalarb{\varphi^{\otimes k}}{\gamma^{(k)} \, \varphi^{\otimes k}}\,.
\end{equation}
\begin{lemma} \label{lemma: relationship between Es}
Let $\gamma^{(k)} \in \mathcal{L}(\mathcal{H}^{(k)})$ satisfy
\begin{equation*}
\gamma^{(k)} \;\geq\; 0 \,,\qquad \tr \gamma^{(k)} \;=\; 1\,.
\end{equation*}
Let $\varphi \in L^2$ satisfy $\norm{\varphi} = 1$. Then
\begin{equation} \label{convergence of k-particle density matrix}
E^{(k)} \;\leq\; k \, E^{(1)}\,.
\end{equation}
\end{lemma}

\begin{proof}
Let $\pb{\Phi_i^{(k)}}_{i \geq 1}$ be an orthonormal basis of $\mathcal{H}^{(k)}$ with 
$\Phi_1^{(k)} = \varphi^{\otimes k}$. Then
\begin{align*}
\scalarb{\varphi^{\otimes k}}{\gamma^{(k)} \, \varphi^{\otimes k}} &\;=\; \sum_{i\geq 1} 
\scalarb{\varphi \otimes \Phi^{(k-1)}_i}{\gamma^{(k)} \, \varphi \otimes \Phi^{(k-1)}_i} - 
\sum_{i\geq 2} \scalarb{\varphi \otimes \Phi^{(k-1)}_i}{\gamma^{(k)} \, \varphi \otimes 
\Phi^{(k-1)}_i}
\\
&\;=\; \scalar{\varphi}{\gamma^{(1)} \, \varphi} - \sum_{i\geq 2} \scalarb{\varphi \otimes 
\Phi^{(k-1)}_i}{\gamma^{(k)} \, \varphi \otimes \Phi^{(k-1)}_i}\,.
\end{align*}
Therefore,
\begin{align*}
&\qquad \scalar{\varphi}{\gamma^{(1)} \, \varphi} - \scalarb{\varphi^{\otimes k}}{\gamma^{(k)} 
\, \varphi^{\otimes k}} \\
&\;=\; \sum_{i\geq 2} \scalarb{\varphi \otimes \Phi^{(k-1)}_i}{\gamma^{(k)} \, \varphi \otimes 
\Phi^{(k-1)}_i}
\\
&\;\leq\; \sum_{i\geq 2} \sum_{j \geq 1} \scalarb{\Phi^{(1)}_j \otimes 
\Phi^{(k-1)}_i}{\gamma^{(k)} \, \Phi^{(1)}_j \otimes \Phi^{(k-1)}_i}
\\
&\;=\;
\sum_{i\geq 1} \sum_{j \geq 1} \scalarb{\Phi^{(1)}_j \otimes \Phi^{(k-1)}_i}{\gamma^{(k)} \, 
\Phi^{(1)}_j \otimes \Phi^{(k-1)}_i} - \sum_{j \geq 1} \scalarb{\Phi^{(1)}_j \otimes 
\varphi^{\otimes (k-1)}}{\gamma^{(k)} \, \Phi^{(1)}_j \otimes \varphi^{\otimes (k-1)}}
\\
&\;=\; 1 - \scalarb{\varphi^{\otimes (k-1)}}{\gamma^{(k-1)} \, \varphi^{\otimes (k-1)}}\,.
\end{align*}
This yields
\begin{equation*}
E^{(k)} \;\leq\; E^{(k-1)} + E^{(1)}\,,
\end{equation*}
and the claim follows.
\end{proof}

\begin{remark}
The bound in \eqref{convergence of k-particle density matrix} is sharp. Indeed, let us suppose 
that $E^{(k)} \;\leq\; k \,f(k)\,  E^{(1)}$ for some function $f$. Then
\begin{equation*}
f(k) \;\geq\; \sup_{\gamma^{(k)}} \, \frac{E^{(k)}}{k E^{(1)}} \;\geq\; \sup_{0 < \alpha < 1} 
\frac{1 - (1 - \alpha)^k}{k \alpha}\;\geq\; \lim_{\alpha \to 0} \frac{1 - (1 - \alpha)^k}{k 
\alpha} \;=\; 1\,,
\end{equation*}
where the second inequality follows by restricting the supremum to product states $\gamma^{(k)} 
= (\ket{\psi} \bra{\psi})^{\otimes k}$ and writing $\alpha = E^{(1)}$.
\end{remark}

The next lemma describes the relationship between $E^{(k)}$ and $R^{(k)}$, where, we recall,
\begin{equation*}
R^{(k)} \;=\; \tr \absb{\gamma^{(k)} - (\ket{\varphi} \bra{\varphi})^{\otimes k}}\,.
\end{equation*}

\begin{lemma} \label{lemma: estimates between R and E}
Let $\gamma^{(k)} \in \mathcal{L}(\mathcal{H}^{(k)})$ be a density matrix and $\varphi \in L^2$ 
satisfy $\norm{\varphi} = 1$. Then
\begin{subequations} \label{R <> E}
\begin{align}
E^{(k)} &\;\leq\; R^{(k)}\,, \label{E < R}
\\
R^{(k)} &\;\leq\; \sqrt{8 \, E^{(k)}}\,. \label{R < E}
\end{align}
\end{subequations}
\end{lemma}
\begin{proof}
It is convenient to introduce the shorthand
\begin{equation*}
p^{(k)} \;\deq\; (\ket{\varphi} \bra{\varphi})^{\otimes k}.
\end{equation*}
Thus,
\begin{equation*}
E^{(k)} \;=\; 1 - \scalarb{\varphi^{\otimes k}}{\gamma^{(k)} \, \varphi^{\otimes k}} \;=\; \tr 
\pb{p^{(k)} -  p^{(k)} \gamma^{(k)}} \;\leq\; \norm{p^{(k)}} \tr \absb{p^{(k)} - \gamma^{(k)}} 
\;=\; R^{(k)}\,,
\end{equation*}
which is \eqref{E < R}. In order to prove \eqref{R < E} it is easiest to use the identity
\begin{equation} \label{Seiringer's estimate}
\tr \absb{p^{(k)} - \gamma^{(k)}} \;=\; 2 \, \normb{p^{(k)} - \gamma^{(k)}}\,,
\end{equation}
valid for any one-dimensional projector $p^{(k)}$ and nonnegative density matrix 
$\gamma^{(k)}$.
This was first observed by Seiringer; see \cite{RodnianskiSchlein2007}. For the convenience of 
the reader we recall the proof of \eqref{Seiringer's estimate}. Let $(\lambda_n)_{n \in \N}$ be 
the sequence of eigenvalues of the trace class operator $A \deq \gamma^{(k)} - p^{(k)}$. Since 
$p^{(k)}$ is a rank one projection, $A$ has at most one negative eigenvalue, say $\lambda_0$.  
Also, $\tr A = 0$ implies that $\sum_{n} \lambda_n = 0$. Thus, $\sum_{n} \abs{\lambda_n} = 2 
\abs{\lambda_0}$, which is \eqref{Seiringer's estimate}.

Now \eqref{Seiringer's estimate} yields
\begin{equation*}
R^{(k)} \;=\; \tr \absb{p^{(k)} - \gamma^{(k)}} \;=\; 2 \, \norm{p^{(k)} - \gamma^{(k)}}  
\;\leq\; 2 \, \sqrt{\tr \pb{p^{(k)} - \gamma^{(k)}}^2}\,.
\end{equation*}
Then \eqref{R < E} follows from
\begin{equation*}
\tr \pb{p^{(k)} - \gamma^{(k)}}^2 \;=\; 1 - 2 \tr \pb{p^{(k)} \gamma^{(k)}} + \tr 
(\gamma^{(k)})^2 \;\leq\; E^{(k)}  - \tr \pb{p^{(k)} \gamma^{(k)}} + 1 \;=\; 2 E^{(k)}\,.
\end{equation*}

Alternatively, one may prove \eqref{R < E} without \eqref{Seiringer's estimate} by using the 
polar decomposition and the Cauchy-Schwarz inequality for Hilbert-Schmidt operators.
\end{proof}

\begin{remark}
Up to constant factors the bounds \eqref{R <> E} are sharp, as the following examples show. 
Here we drop the irrelevant index $k$. Consider first
\begin{equation*}
\varphi \;=\;
\begin{pmatrix}
1 \\0
\end{pmatrix}\,,
\qquad
\gamma \;=\; \begin{pmatrix}
1 - a & 0
\\
0 & a
\end{pmatrix}\,,
\end{equation*}
where $0 \leq a \leq 1$. As above we set $p \deq \ket{\varphi} \bra{\varphi}$.
One finds
\begin{equation*}
E \;=\; 1 - \scalar{\varphi}{\gamma \, \varphi} \;=\; a \,,\qquad R \;=\; \tr \abs{p - \gamma} 
\;=\; 2a\,,
\end{equation*}
so that \eqref{E < R} is sharp up to a constant factor.

It is not hard to see that if $\gamma$ and $p$ commute then \eqref{R < E} can be replaced with 
the stronger bound $R \lesssim E$. In order to show that in general \eqref{R < E} is sharp up 
to a constant factor, consider
\begin{equation*}
\varphi \;=\;
\begin{pmatrix}
1 \\0
\end{pmatrix}\,,
\qquad
\gamma \;=\; \begin{pmatrix}
1 - a & \sqrt{a - a^2}
\\
\sqrt{a - a^2} & a
\end{pmatrix}\,,
\end{equation*}
where $0 \leq a \leq 1$. One readily sees that $\gamma$ is a density matrix (in fact, a 
one-dimensional projector). A short calculation yields
\begin{equation*}
E \;=\; 1 - \scalar{\varphi}{\gamma \, \varphi} \;=\; a
\end{equation*}
as well as
\begin{equation*}
\tr \absb{\gamma (1 - p)} \;=\; \sqrt{a}\,.
\end{equation*}
Using
\begin{equation*}
\tr \absb{\gamma(1 - p)} \;=\; \tr \absb{\gamma - p + p - \gamma p } \;\leq\; 2 \, \tr \abs{p - 
\gamma}
\end{equation*}
we therefore find
\begin{equation*}
R \;=\; \tr \abs{p - \gamma} \;\geq\; \frac{\sqrt{a}}{2} \;=\; \frac{\sqrt{E}}{2}\,,
\end{equation*}
as desired.
\end{remark}

\section{Convergence for $L^2$-type singularities} \label{L2 potentials}
This section is devoted to the case $w \in L^2 + L^\infty$.
\subsection{Outline and main result}
Our method relies on controlling the quantity
\begin{equation}
\alpha_N(t) \;\deq\; E^{(1)}_N(t)\,.
\end{equation}
To this end, we derive an estimate of the form
\begin{equation} \label{main estimate for alpha'}
\dot{\alpha}_N(t) \;\leq\; A_N(t) + B_N(t) \, \alpha_N(t)\,,
\end{equation}
which, by Gr\"onwall's lemma, implies
\begin{equation} \label{Gronwall}
\alpha_N(t) \;\leq\; \alpha_N(0) \, \ee^{\int_0^t B_N} + \int_0^t A_N(s) \, \ee^{\int_s^t B_N} 
\, \dd s\,.
\end{equation}
In order to show \eqref{main estimate for alpha'}, we differentiate $\alpha_N(t)$ and note that 
all terms arising from the one-particle Hamiltonian vanish. We control the remaining terms by 
introducing the time-dependent orthogonal projections
\begin{equation*}
p(t) \;\deq\; \ket{\varphi(t)} \bra{\varphi(t)} \,,\qquad q(t) \;\deq\; \umat - p(t)\,.
\end{equation*}
We then partition $\umat = p(t) + q(t)$ appropriately and use the following heuristics for 
controlling the terms that arise in this manner. Factors $p(t)$ are used to control 
singularities of $w$ by exploiting the smoothness of the Hartree wave function $\varphi(t)$.  
Factors $q(t)$ are expected to yield something small, i.e.\ proportional to $\alpha_N(t)$, in 
accordance with the identity $\alpha_N(t) = \scalar{\Psi_N(t)}{q_1(t)\Psi_N(t)}$.

For the following it is convenient to rewrite the Hamiltonian \eqref{mean-field Hamiltonian in 
introduction} as
\begin{equation} \label{mean-field Hamiltonian}
H_N \;=\; \sum_{i = 1}^N h_i + \frac{1}{N} \sum_{1 \leq i < j \leq N} W_{ij} \;\eqd\; H_N^0 + 
H_N^W\,,
\end{equation}
where $W_{ij} \;\deq\; w(x_i - x_j)$. We may now list our assumptions.
\begin{itemize}
\item[(A1)] The one-particle Hamiltonian $h$ is self-adjoint and bounded from below. Without 
loss of generality we assume that $h \geq 0$. We define the Hilbert space $X_N \;=\; 
\mathcal{Q}(H_N^0)$ as the form domain of $H_N^0$ with norm
\begin{equation*}
\norm{\Psi}_{X_N} \;\deq\; \normb{(1 + H_N^0)^{1/2} \Psi}\,.
\end{equation*}
\item[(A2)] The Hamiltonian \eqref{mean-field Hamiltonian} is self-adjoint and bounded from 
below. We also assume that $\mathcal{Q}(H_N) \subset X_N$.
\item[(A3)] The interaction potential $w$ is a real and even function satisfying $w \in L^{p_1} 
+ L^{p_2}$, where $2 \leq p_1 \leq p_2 \leq \infty$.
\item[(A4)] The solution $\varphi(\cdot)$ of \eqref{hartree equation} satisfies 
\begin{equation*}
\varphi(\cdot) \in C(\R; X_1 \cap L^{q_1}) \cap C^1(\R; X_1^*)\,,
\end{equation*}
where $2 \leq q_2 \leq q_1 \leq \infty$ are defined through
\begin{equation}
\frac{1}{2} \;=\; \frac{1}{p_i} + \frac{1}{q_i}\,, \qquad i = 1,2\,.
\end{equation}
Here $X_1^*$ denotes the dual space of $X_1$, i.e.\ the closure of $L^2$ under the norm 
$\norm{\varphi}_{X_1^*} \deq \norm{(\umat + h)^{-1/2} \varphi}$.
\end{itemize}

We now state our main result.  \begin{theorem} \label{theorem for L^2 potentials}
Let $\Psi_{N,0} \in \mathcal{Q}(H_N)$ satisfy $\norm{\Psi_{N,0}} = 1$, and $\varphi_0 \in X_1 
\cap L^{q_1}$ satisfy $\norm{\varphi_0} = 1$. Assume that Assumptions {\rm (A1)} -- {\rm (A4)} 
hold. Then
\begin{equation*}
\alpha_N(t) \;\leq\; \pbb{\alpha_N(0) + \frac{1}{N}} \,\ee^{\phi(t)}\,,
\end{equation*}
where
\begin{equation*}
\phi(t) \;\deq\; 32 \norm{w}_{L^{p_1} + L^{p_2}} \int_0^t \dd s \; \pb{\norm{\varphi(s)}_{q_1} 
+ \norm{\varphi(s)}_{q_2}}\,.
\end{equation*}
\end{theorem}

We may combine this result with the observations of Section \ref{measures of convergence}.
\begin{corollary} \label{corollary of main theorem}
Let the sequence $\Psi_{N,0} \in \mathcal{Q}(H_N)$, $N \in \N$, satisfy the assumptions of 
Theorem \ref{theorem for L^2 potentials} as well as
\begin{equation*}
E_N^{(1)}(0) \;\lesssim\; \frac{1}{N}\,.
\end{equation*}
Then  we have
\begin{equation*}
E^{(k)}_N(t) \;\lesssim\; \frac{k}{N} \, \ee^{\phi(t)}\,, \qquad R^{(k)}_N(t) \;\lesssim\; 
\sqrt{\frac{k}{N}} \, \ee^{\phi(t)/2}\,.
\end{equation*}
\end{corollary}

\begin{remark}
Corollary \ref{corollary of main theorem} implies that we can control the condensation of $k = 
o(N)$ particles.
\end{remark}
\begin{remark}
Assumption (A3) allows for singularities in $w$ up to, but not including, the type 
$\abs{x}^{-3/2}$ in three dimensions. In the next section we treat a larger class of 
interaction potentials.
\end{remark}
\begin{remark}
Assumption (A4) is typically verified by solving the Hartree equation in a Sobolev space of 
high index (see e.g.\ Section \ref{example: semirelativistic kinetic energy}). Instead of 
requiring a global-in-time solution $\varphi(\cdot)$, it is enough to have a local-in-time 
solution on $[0,T)$ for some $T > 0$.  \end{remark}
\begin{remark} \label{remark: time-dependent potentials}
If $\sup_t \phi(t) < \infty$, or in other words if $\norm{\varphi(t)}_{q_1}$ and 
$\norm{\varphi(t)}_{q_2}$ are integrable in $t$ over $\R$, then all estimates are uniform in 
time.
This describes a scattering regime where the time evolution is asymptotically free for large 
times. Such an integrability condition requires large exponents $q_i$, which translates to 
small exponents $p_i$, i.e.\ an interaction potential with strong decay.
\end{remark}
\begin{remark}
The result easily extends to time-dependent one-particle Hamiltonians $h \equiv h(t)$. Replace 
(A1) and (A2) with
\begin{itemize}
\item[(A1')]
The Hamiltonian $h(t)$ is self-adjoint and bounded from below. We assume that there is an 
operator $h_0 \geq 0$ that  such that $0 \leq h(t) \leq h_0$ for all $t$. Define the Hilbert 
space $X_N \;=\; \mathcal{Q}\pb{\sum_i (h_0)_i}$ as in (A1).
\item[(A2')]
The Hamiltonian $H_N(t)$ is self-adjoint and bounded from below.
We assume that $\mathcal{Q}(H_N(t)) \subset X_N$ for all $t$.
We also assume that the $N$-body propagator $U_N(t,s)$, defined by
\begin{equation*}
\ii \partial_t U_N(t,s) = H_N(t) U_N(t,s) \,, \qquad U_N(s,s) = \umat\,,
\end{equation*}
exists and satisfies $U_N(t,0)\Psi_{N,0} \in \mathcal{Q}(H_N(t))$ for all $t$.
\end{itemize}
It is then straightforward that Theorem \ref{theorem for L^2 potentials} holds with the same 
proof.
\end{remark}
\begin{remark} \label{remark: alternative assumptions for Hardy}
In some cases (see e.g.\ Section \ref{example: nonrelativistic kinetic energy} below) it is 
convenient to modify the assumptions as follows. Replace (A3) and (A4) with
\begin{itemize}
\item[(A3')]
The interaction potential $w$ is a real and even function satisfying
\begin{equation} \label{bound on square of w with Hardy}
\normb{w^2 * \abs{\varphi}^2}_\infty \;\leq\; K \, \norm{\varphi}_{X_1}^2
\end{equation}
for some constant $K > 0$. Without loss of generality we assume that $K \geq 1$.
\item[(A4')]
The solution $\varphi(\cdot)$ of \eqref{hartree equation} satisfies \begin{equation*}
\varphi(\cdot) \in C(\R; X_1) \cap C^1(\R; X_1^*)\,.
\end{equation*}
\end{itemize}
Then Theorem \ref{theorem for L^2 potentials} and Corollary \ref{corollary of main theorem} 
hold with
\begin{equation*}
\phi(t) \;=\; 32 K \int_0^t \dd s \; \norm{\varphi(s)}^2_{X_1}\,.
\end{equation*}
The proof remains virtually unchanged. One replaces \eqref{bound on Wtilde} with \eqref{bound 
on square of w with Hardy}, as well as \eqref{bound of convolution with w} with
\begin{equation*}
\normb{w * \abs{\varphi}^2}_\infty \;\leq\; 2 K \, \norm{\varphi}_{X_1}^2\,,
\end{equation*}
which is an easy consequence of \eqref{bound on square of w with Hardy}.
\end{remark}

\subsection{Examples}
We list two examples of systems satisfying the assumptions of Theorem \ref{theorem for L^2 
potentials}.

\subsubsection{Particles in a trap} \label{example: nonrelativistic kinetic energy}
Consider nonrelativistic particles in $\R^3$ confined by a strong trapping potential. The 
particles interact by means of the Coulomb potential: $w(x) = \lambda \abs{x}^{-1}$, where 
$\lambda \in \R$. The one-particle Hamiltonian is of the form $h = -\Delta + v$, where $v$ is a 
measurable function on $\R^3$. Decompose $v$ into its positive and negative parts: $v = v_+ - 
v_-$, where $v_+, v_- \geq 0$. We assume that $v_+ \in L^1_{\mathrm{loc}}$ and that $v_-$ is 
$-\Delta$-form bounded with relative bound less than one, i.e.\ there are constants $0 \leq a < 
1$ and $0 \leq b < \infty$ such that
\begin{equation} \label{bound on negative part of interaction potential}
\scalar{\varphi}{v_- \varphi} \;\leq\; a \scalar{\varphi}{-\Delta \varphi} + b 
\scalar{\varphi}{\varphi}\,.
\end{equation}
Thus $h + b\umat$ is positive, and it is not hard to see that $h$ is essentially self-adjoint 
on $C^\infty_c(\R^3)$. This follows by density and a standard argument using Riesz's 
representation theorem to show that the equation $(h + (b + 1)\umat) \varphi = f$ has a unique 
solution $\varphi \in \{\varphi \in L^2 \,:\, h \varphi \in L^2\}$ for each $f \in L^2$.

It is now easy to see that Assumptions (A1) and (A2) hold with the one-particle Hamiltonian $h 
+ c \umat$ for some $c > 0$. Let us assume without loss of generality that $c = 0$. Next, we 
verify Assumptions (A3') and (A4') (see Remark \ref{remark: alternative assumptions for 
Hardy}). We find
\begin{equation*}
\normb{w^2 * \abs{\varphi}^2}_\infty \;=\; \sup_{x} \absbb{\int \dd y \; 
\frac{\lambda^2}{\abs{x - y}^2} \abs{\varphi(y)}^2} \;\lesssim\; \scalar{\varphi}{-\Delta 
\varphi} \;\lesssim\; \scalar{\varphi}{h \varphi} + \scalar{\varphi}{\varphi} \;=\; 
\norm{\varphi}_{X_1}^2\,,
\end{equation*}
where the second step follows from Hardy's inequality and translation invariance of $\Delta$, 
and the third step is a simple consequence of \eqref{bound on negative part of interaction 
potential}. This proves (A3').

Next, take $\varphi_0 \in X_1$. By standard methods (see e.g.\ the presentation of 
\cite{Lenzmann2007}) one finds that (A4') holds. Moreover, the mass $\norm{\varphi(t)}^2$ and 
the energy
\begin{equation*}
E^\varphi(t) \;=\; \qbb{\scalar{\varphi}{h \varphi} + \frac{1}{2} \int \dd x \, \dd y \; w(x - 
y) \abs{\varphi(x)}^2 \abs{\varphi(y)}^2} \biggr|_t
\end{equation*}
are conserved under time evolution. Using the identity $\abs{x}^{-1} \leq \umat_{\{\abs{x} \leq 
\epsilon\}} \epsilon \abs{x}^{-2} + \umat_{\{\abs{x} > \epsilon\}} \epsilon^{-1}$ and Hardy's 
inequality one sees that
\begin{equation*}
\norm{\varphi(t)}_{X_1}^2 \;\lesssim\; E^\varphi(t) + \norm{\varphi(t)}^2\,,
\end{equation*}
and therefore $\norm{\varphi(t)}_{X_1} \leq C$ for all $t$. We conclude: Theorem \ref{theorem 
for L^2 potentials} holds with $\phi(t) = C t$. More generally, the preceding discussion holds 
for interaction potentials $w \in L^3_w + L^\infty$, where $L^p_w$ denotes the weak $L^p$ space 
(see e.g.\ \cite{ReedSimonII}). This follows from a short computation using 
symmetric-decreasing rearrangements; we omit further details. This example generalizes the 
results of \cite{ErdosYau2001}, \cite{RodnianskiSchlein2007} and 
\cite{FrohlichKnowlesSchwarz2009}.

\subsubsection{A boson star} \label{example: semirelativistic kinetic energy}
Consider semirelativistic particles in $\R^3$ whose one-particle Hamiltonian is given by $h = 
\sqrt{\umat - \Delta}$. The particles interact by means of a Coulomb potential: $w(x) = \lambda 
\abs{x}^{-1}$. We impose the condition $\lambda > -4/\pi$. This condition is necessary for both 
the stability of the $N$-body problem (i.e.\ Assumption (A2)) and the global well-posedness of 
the Hartree equation. See \cite{LiebYau1987, Lenzmann2007} for details. It is well known that 
Assumptions (A1) and (A2) hold in this case.

In order to show (A4) we need some regularity of $\varphi(\cdot)$. To this end, let $s>1$ and 
take $\varphi_0 \in H^s$. Theorem 3 of \cite{Lenzmann2007} implies that \eqref{hartree 
equation} has a unique global solution in $H^s$. Therefore Sobolev's inequality implies that 
(A4) holds with
\begin{equation*}
\frac{1}{q_1} \;=\; \frac{1}{2} - \frac{s}{3}\,.
\end{equation*}
Thus $q_1 > 6$, and (A3) holds with appropriately chosen values of $p_1, p_2$. We conclude: 
Theorem \ref{theorem for L^2 potentials} holds for some continuous function $\phi(t)$. (In 
fact, as shown in \cite{Lenzmann2007}, one has the bound $\phi(t) \lesssim \ee^{Ct}$.) This 
example generalizes the result of \cite{ElgartSchlein2007}.

\subsection{Proof of Theorem \ref{theorem for L^2 potentials}}
\subsubsection{A family of projectors}
Define the time-dependent projectors
\begin{equation*}
p(t) \;\deq\; \ket{\varphi(t)} \bra{\varphi(t)} \,,\qquad q(t) \;\deq\; \umat - p(t)\,.
\end{equation*}
Write
\begin{equation} \label{partition of 1}
\umat \;=\; (p_1 + q_1) \cdots (p_N + q_N)
\end{equation}
and define $P_k$, for $k = 0, \dots, N$, as the term obtained by multiplying out 
\eqref{partition of 1} and selecting all summands containing $k$ factors $q$. In other words,
\begin{equation}
P_k \;=\; \sum_{\substack{a \in \{0,1\}^N \,:\\ \sum_i a_i = k}} \, \prod_{i = 1}^N p_i^{1 - 
a_i} q_i^{a_i}\,.
\end{equation}
If $k \neq \{0, \dots, N\}$ we set $P_k \;=\; 0$.  It is easy to see that the following 
properties hold:
\begin{enumerate}
\item $P_k$ is an orthogonal projector,
\item $P_k P_l \;=\; \delta_{kl} P_k$\,,
\item $\sum_k P_k \;=\; \umat$\,.
\end{enumerate}

Next, for any function $f : \{0, \dots, N\} \to \C$ we define the operator
\begin{equation}
\widehat{f} \;\deq\; \sum_{k} f(k) P_k\,.
\end{equation}
It follows immediately that
\begin{equation*}
\widehat{f} \widehat{g} \;=\; \widehat{fg}\,,
\end{equation*}
and that $\widehat{f}$ commutes with $p_i$ and $P_k$.
We shall often make use of the functions
\begin{equation*}
m(k) \;\deq\; \frac{k}{N} \,,\qquad n(k) \;\deq\; \sqrt{\frac{k}{N}}\,.
\end{equation*}
We have the relation
\begin{equation} \label{replacing q with m 0}
\frac{1}{N} \sum_i q_i \;=\; \frac{1}{N} \sum_k \sum_i q_i P_k \;=\; \frac{1}{N} \sum_k k P_k 
\;=\; \widehat{m}\,.
\end{equation}
Thus, by symmetry of $\Psi$, we get
\begin{equation} \label{two forms of alpha}
\alpha \;=\; \scalar{\Psi}{q_1 \, \Psi} \;=\; \scalar{\Psi}{\widehat{m} \, \Psi}\,.
\end{equation}
The correspondence $q_1 \sim \widehat{m}$ of \eqref{replacing q with m 0} yields the following 
useful bounds.
\begin{lemma} \label{replacing q's with n's}
For any nonnegative function $f : \{0,\dots,N\} \to [0,\infty)$ we have
\begin{align}
\scalarb{\Psi}{\widehat{f} q_1 \Psi} &\;=\; \scalarb{\Psi}{\widehat{f} \, \widehat{m} \Psi}\,, 
\label{replacing q with m 1}
\\
\scalarb{\Psi}{\widehat{f} q_1 q_2 \Psi} &\;\leq\; \frac{N}{N-1} \scalarb{\Psi}{\widehat{f} 
\,\widehat{m}^2 \Psi}\,.\label{replacing q with m 2}
\end{align}
\end{lemma}
\begin{proof}
The proof of \eqref{replacing q with m 1} is an immediate consequence of \eqref{replacing q 
with m 0}. In order to prove \eqref{replacing q with m 2} we write, using symmetry of $\Psi$ as 
well as \eqref{replacing q with m 0},
\begin{multline*}
\scalarb{\Psi}{\widehat{f} \, q_1 q_2 \Psi} \;=\; \frac{1}{N(N-1)} \sum_{i \neq j} 
\scalarb{\Psi}{\widehat{f} \, q_i q_j \Psi}
\\
\;\leq\; \frac{1}{N(N-1)} \sum_{i, j} \scalarb{\Psi}{\widehat{f} \, q_i q_j \Psi} \;=\; 
\frac{N}{N-1} \scalarb{\Psi}{\widehat{f} \, \widehat{m}^2 \Psi}\,,
\end{multline*}
which is the claim.
\end{proof}

Next, we introduce the shift operation $\tau_n$, $n \in \Z$, defined on functions $f$ through
\begin{equation}
(\tau_n f)(k) \;\deq\; f(k + n)\,.
\end{equation}
Its usefulness for our purposes is encapsulated by the following lemma.
\begin{lemma} \label{pulling projectors through}
Let $r \geq 1$ and $A$ be an operator on $\mathcal{H}^{(r)}$. Let $Q_i$, $i = 1,2$, be two 
projectors of the form
\begin{equation*}
Q_i \;=\; \#_1 \cdots \#_r\,,
\end{equation*}
where each $\#$ stands for either $p$ or $q$. Then
\begin{equation*}
Q_1 A_{1\dots r} \widehat{f} Q_2 \;=\; Q_1 \widehat{\tau_n f} A_{1 \dots r}  Q_2\,,
\end{equation*}
where $n = n_2 - n_1$ and $n_i$ is the number of factors $q$ in $Q_i$.
\end{lemma}
\begin{proof}
Define
\begin{equation*}
P^r_k \;\deq\; \sum_{\substack{a \in \{0,1\}^{N - r} \\ \sum_i a_i = k}} \, \prod_{i = r+1}^N 
p_i^{1 - a_i} q_i^{a_i}\,.
\end{equation*}
Then,
\begin{equation*}
Q_i \widehat{f} \;=\; \sum_k f(k) \, Q_i P_k \;=\; \sum_k f(k) \, Q_i P^r_{k - n_i} \;=\; 
\sum_k f(k + n_i) \, Q_i P_k^r\,.
\end{equation*}
The claim follows from the fact that $P_k^r$ commutes with $A_{1\dots r}$.
\end{proof}

\subsubsection{A bound on $\dot{\alpha}$}
Let us abbreviate
\begin{equation*}
W^\varphi \;\deq\; w * \abs{\varphi}^2\,.
\end{equation*}
From (A3) and (A4) we find $W^\varphi \in L^\infty$ (see \eqref{bound of convolution with w} 
below).
Then $\ii \partial_t \varphi \;=\; (h + W^\varphi) \varphi$, where $h + W^\varphi \in 
\mathcal{L}(X_1; X_1^*)$. Thus, for any $\psi \in X_1$ independent of $t$ we have
\begin{equation*}
\ii \partial_t \scalar{\psi}{p \, \psi} \;=\; \scalar{\psi}{\com{h + W^\varphi}{p} \psi}\,.
\end{equation*}
On the other hand, it is easy to see from (A3) and (A4) that $\widehat{m} \Psi \in 
\mathcal{Q}(H)$.
Combining these observations, and noting that $\Psi \in \mathcal{Q}(H) \subset X$ by (A2),
we see that $\alpha$ is differentiable in $t$ with derivative
\begin{equation*}
\dot{\alpha} \;=\; \ii \scalarb{\Psi}{\comb{H - H^\varphi}{\widehat{m}} \Psi}\,,
\end{equation*}
where $H^\varphi \deq \sum_i \p{h_i + W_i^\varphi}$. Thus,
\begin{equation*}
\dot{\alpha} \;=\; \ii \scalarbb{\Psi}{\combb{\frac{1}{N}\sum_{i<j} W_{ij} - \sum_i 
W_i^\varphi}{\widehat{m}}\Psi}\,.
\end{equation*}
By symmetry of $\Psi$ and $\widehat{m}$ we get
\begin{equation} \label{derivative of alpha}
\dot{\alpha} \;=\; \frac{\ii}{2} \scalarb{\Psi}{\comb{(N-1) W_{12} - N W^\varphi_1 - N 
W^\varphi_2}{\widehat{m}}\Psi}\,.
\end{equation}
In order to estimate the right-hand side, we introduce
\begin{equation*}
\umat \;=\; (p_1 + q_1)(p_2 + q_2)
\end{equation*}
on both sides of the commutator in \eqref{derivative of alpha}. Of the sixteen resulting terms 
only three different types survive:
\begin{align*}
&\frac{\ii}{2} \scalarb{\Psi}{p_1 p_2 \comb{(N-1) W_{12} - N W^\varphi_1 - N 
W^\varphi_2}{\widehat{m}} q_1 p_2\Psi} & &\mathrm{(I)}
\\
&\frac{\ii}{2} \scalarb{\Psi}{q_1 p_2 \comb{(N-1) W_{12} - N W^\varphi_1 - N 
W^\varphi_2}{\widehat{m}} q_1 q_2\Psi} & &\mathrm{(II)}
\\
&\frac{\ii}{2} \scalarb{\Psi}{p_1 p_2 \comb{(N-1) W_{12} - N W^\varphi_1 - N 
W^\varphi_2}{\widehat{m}} q_1 q_2\Psi} & &\mathrm{(III)}\,.
\end{align*}
Indeed, Lemma \ref{pulling projectors through} implies that terms with the same number of 
factors $q$ on the left and on the right vanish. What remains is
\begin{equation*}
\dot{\alpha} \;=\; 2 \mathrm{(I)} + 2 \mathrm{(II)} + \mathrm{(III)} + \text{ complex 
conjugate}\,.
\end{equation*}
The remainder of the proof consists in estimating each term.

\step{Term $\mathrm{(I)}$.}
First, we remark that
\begin{equation} \label{W sandwiched between p's}
p_2 W_{12} p_2 \;=\; p_2 W_1^\varphi\,.
\end{equation}
This is easiest to see using operator kernels (we drop the trivial indices $x_3, y_3, \dots, 
x_N, y_N$):
\begin{align*}
(p_2 W_{12} p_2)(x_1, x_2; y_1, y_2) &\;=\; \int \dd z \; \varphi(x_2) \, \ol{\varphi}(z) \, 
w(x_1 - z) \, \delta(x_1 - y_1) \, \varphi(z) \, \ol{\varphi}(y_2)
\\
&\;=\; \varphi(x_2) \, \ol{\varphi}(y_2) \, \delta(x_1 - y_1) \, (w * \abs{\varphi}^2)(x_1)\,.
\end{align*}
Therefore,
\begin{equation*}
\mathrm{(I)} \;=\; \frac{\ii}{2} \scalarb{\Psi}{p_1 p_2 \comb{(N-1) W_1^\varphi - N 
W^\varphi_1}{\widehat{m}} q_1 p_2\Psi} \;=\; \frac{- \ii}{2} \scalarb{\Psi}{p_1 p_2 
\comb{W^\varphi_1}{\widehat{m}} q_1 p_2\Psi}\,.
\end{equation*}
Using Lemma \ref{pulling projectors through} we find
\begin{equation*}
\mathrm{(I)} \;=\; \frac{-\ii}{2} \scalarb{\Psi}{p_1 p_2 W^\varphi_1 \pb{\widehat{m} - 
\widehat{\tau_{-1} m}} q_1 p_2\Psi}
\;=\; \frac{-\ii}{2N} \scalarb{\Psi}{p_1 p_2 W^\varphi_1 q_1 p_2\Psi}\,.
\end{equation*}
This gives
\begin{equation*}
\absb{\mathrm{(I)}} \;\leq\; \frac{1}{2N} \, \norm{W^\varphi}_\infty \;=\; \frac{1}{2N} 
\normb{w * \abs{\varphi}^2}_\infty\,.  \end{equation*}
By (A3), we may write
\begin{equation} \label{decomposition of w}
w \;=\; w^{(1)} + w^{(2)}\,,\qquad w^{(i)} \;\in\; L^{p_i}\,.
\end{equation}
By Young's inequality,
\begin{equation*}
\normb{w^{(i)} * \abs{\varphi}^2}_\infty \;\leq\; \norm{w^{(i)}}_{p_i} 
\norm{\varphi}_{r_i}^2\,,
\end{equation*}
where $r_1, r_2$ are defined through
\begin{equation}
1 \;=\; \frac{1}{p_i} + \frac{2}{r_i}\,.
\end{equation}
Therefore,
\begin{equation*}
\normb{w * \abs{\varphi}^2}_\infty \;\leq\; \norm{w^{(1)}}_{p_1} \norm{\varphi}^2_{r_1} + 
\norm{w^{(1)}}_{p_2} \norm{\varphi}^2_{r_2} \;\leq\; \pb{\norm{w^{(1)}}_{p_1} + 
\norm{w^{(2)}}_{p_2}} \pb{\norm{\varphi}_{r_1} + \norm{\varphi}_{r_2}}^2\,.
\end{equation*}
Taking the infimum over all decompositions \eqref{decomposition of w} yields
\begin{equation} \label{bound of convolution with w}
\norm{W^\varphi}_\infty \;=\; \normb{w * \abs{\varphi}^2}_\infty \;\leq\; \norm{w}_{L^{p_1} + 
L^{p_2}} \pb{\norm{\varphi}_{r_1} + \norm{\varphi}_{r_2}}^2\,.
\end{equation}
Note that (A3) and (A4) imply
\begin{equation} \label{r < q}
2 \;\leq\; r_i \;\leq\; q_1\,,
\end{equation}
so that the right-hand side of \eqref{bound of convolution with w} is finite.
Summarizing,
\begin{equation}
\absb{\mathrm{(I)}} \;\leq\; \frac{1}{2N} \norm{w}_{L^{p_1} + L^{p_2}} \pb{\norm{\varphi}_{r_1} 
+ \norm{\varphi}_{r_2}}^2\,.
\end{equation}

\step{Term $\mathrm{(II)}$.}
Applying Lemma \ref{pulling projectors through} to (II) yields
\begin{align*}
\mathrm{(II)} &\;=\; \frac{\ii}{2} \scalarb{\Psi}{q_1 p_2 \pb{(N-1) W_{12} - N W^\varphi_2} 
\pb{\widehat{m} - \widehat{\tau_{-1} m}} q_1 q_2\Psi}
\\
&\;=\; \frac{\ii}{2} \scalarbb{\Psi}{q_1 p_2 \pbb{\frac{N-1}{N} W_{12} -  W^\varphi_2} q_1 
q_2\Psi}\,,
\end{align*}
so that
\begin{equation} \label{splitting of (II)}
\absb{\mathrm{(II)}} \;\leq\; \frac{1}{2} \absb{\scalarb{\Psi}{q_1 p_2 W_{12} q_1 q_2 \Psi}} + 
\frac{1}{2} \absb{\scalarb{\Psi}{q_1 p_2 W_2^\varphi q_1 q_2 \Psi}}\,.
\end{equation}
The second term of \eqref{splitting of (II)} is bounded by
\begin{equation*}
\frac{1}{2} \norm{W^\varphi}_\infty \, \norm{q_1 \Psi}^2 \;\leq\; \frac{1}{2} \norm{w}_{L^{p_1} 
+ L^{p_2}} \pb{\norm{\varphi}_{r_1} + \norm{\varphi}_{r_2}}^2 \, \alpha\,,
\end{equation*}
where we used the bound \eqref{bound of convolution with w} as well as \eqref{two forms of 
alpha}.

The first term of \eqref{splitting of (II)} is bounded using Cauchy-Schwarz by
\begin{equation*}
\frac{1}{2} \sqrt{\scalarb{\Psi}{q_1 p_2 W_{12}^2 p_2 q_1 \Psi}} \sqrt{\scalar{\Psi}{q_1 q_2 
\Psi}} \;=\; \frac{1}{2} \sqrt{\scalarb{\Psi}{q_1 p_2 \pb{w^2 * \abs{\varphi}^2}_1 p_2 q_1 
\Psi}} \sqrt{\scalar{\Psi}{q_1 q_2 \Psi}}\,.
\end{equation*}
This follows by applying \eqref{W sandwiched between p's} to $W^2$. Thus we get the bound
\begin{equation*}
\frac{1}{2} \norm{q_1 \Psi}^2 \sqrt{\normb{w^2 * \abs{\varphi}^2}_\infty} \;=\; \frac{1}{2} 
\alpha \sqrt{\normb{w^2 * \abs{\varphi}^2}_\infty}\,.
\end{equation*}
We now proceed as above. Using the decomposition \eqref{decomposition of w} we get
\begin{equation*}
\normb{w^2 * \abs{\varphi}^2}_\infty \;\leq\; 2 \normb{(w^{(1)})^2 * \abs{\varphi}^2}_\infty + 
2 \normb{(w^{(2)})^2 * \abs{\varphi}^2}_\infty\,.
\end{equation*}
Then Young's inequality gives
\begin{equation*}
\normb{(w^{(i)})^2 * \abs{\varphi}^2}_\infty \;\leq\; \normb{w^{(i)}}^2_{p_i} 
\norm{\varphi}^2_{q_i}\,,
\end{equation*}
which implies that
\begin{equation} \label{bound on Wtilde}
\normb{w^2 * \abs{\varphi}^2}_\infty \;\leq\; 2 \norm{w}_{L^{p_1} + L^{p_2}}^2 
\pb{\norm{\varphi}_{q_1} + \norm{\varphi}_{q_2}}^2\,.
\end{equation}
Putting all of this together we get
\begin{equation*}
\absb{\mathrm{(II)}} \;\leq\; \frac{1}{2} \norm{w}_{L^{p_1} + L^{p_2}} \qB{\sqrt{2} 
\pb{\norm{\varphi}_{q_1} + \norm{\varphi}_{q_2}} + \pb{\norm{\varphi}_{r_1} + 
\norm{\varphi}_{r_2}}^2} \, \alpha\,.
\end{equation*}

\step{Term $\mathrm{(III)}$.}
The final term (III) is equal to
\begin{multline*}
\frac{\ii}{2} \scalarb{\Psi}{p_1 p_2 \comb{(N-1) W_{12}}{\widehat{m}} q_1 q_2\Psi} \;=\; 
\frac{\ii}{2} \scalarb{\Psi}{p_1 p_2 (N-1) W_{12}\pb{\widehat{m} - \widehat{\tau_{-2} m}} q_1 
q_2\Psi}
\\
\;=\; \ii \frac{N-1}{N} \scalarb{\Psi}{p_1 p_2 W_{12} q_1 q_2\Psi}\,,
\end{multline*}
where we used Lemma \ref{pulling projectors through}. Next, we note that, on the range of 
$q_1$, the operator $\widehat{n}^{-1}$ is well-defined and bounded. Thus (III) is equal to
\begin{equation*}
\ii \frac{N-1}{N} \scalarb{\Psi}{p_1 p_2 W_{12} \, \widehat{n} \, \widehat{n}^{-1} q_1 q_2\Psi} 
\;=\; \ii \frac{N-1}{N} \scalarb{\Psi}{p_1 p_2 \, \widehat{\tau_2 n} \, W_{12} \, 
\widehat{n}^{-1} q_1 q_2\Psi}\,,
\end{equation*}
where we used Lemma \ref{pulling projectors through} again. We now use Cauchy-Schwarz to get
\begin{align*}
\absb{\mathrm{(III)}} &\;\leq\; \sqrt{\scalarb{\Psi}{p_1 p_2 \,\widehat{\tau_2 n} \,W_{12}^2 
\,\widehat{\tau_2 n} \,p_1 p_2 \Psi}} \sqrt{\scalarb{\Psi}{\widehat{n}^{-2} q_1 q_2 \Psi}}
\\
&\;=\; \sqrt{\scalarb{\Psi}{p_1 p_2 \,\widehat{\tau_2 n} \, \pb{w^2 * \abs{\varphi}^2}_1 
\,\widehat{\tau_2 n} \,p_1 p_2 \Psi}} \sqrt{\scalarb{\Psi}{\widehat{m}^{-1} q_1 q_2 \Psi}}
\\
&\;\leq\; \sqrt{\normb{w^2 * \abs{\varphi}^2}_\infty} \, \norm{\widehat{\tau_2 n} \Psi} 
\sqrt{\frac{N}{N-1}} \sqrt{\scalar{\Psi}{\widehat{m} \Psi}}
\\
&\;=\; \sqrt{\normb{w^2 * \abs{\varphi}^2}_\infty}  \sqrt{\frac{N}{N-1}} 
\sqrt{\scalarb{\Psi}{\widehat{\tau_2 m} \Psi}} \sqrt{\alpha}
\\
&\;=\; \sqrt{\normb{w^2 * \abs{\varphi}^2}_\infty}  \sqrt{\frac{N}{N-1}} 
\sqrt{\scalarb{\Psi}{\widehat{m} \Psi} + \frac{2}{N}} \, \sqrt{\alpha}
\\
&\;\leq\; \sqrt{\normb{w^2 * \abs{\varphi}^2}_\infty}  \sqrt{\frac{N}{N-1}} \pbb{\alpha + 
\sqrt{\frac{2 \alpha}{N}}}
\\
&\;\leq\; \sqrt{\normb{w^2 * \abs{\varphi}^2}_\infty}  \sqrt{\frac{N}{N-1}} \, 2 \pbb{\alpha + 
\frac{1}{N}}\,.
\end{align*}
Using the estimate \eqref{bound on Wtilde} we get finally
\begin{equation*}
\absb{\mathrm{(III)}} \;\leq\; 2 \sqrt{2} \norm{w}_{L^{p_1} + L^{p_2}} \pb{\norm{\varphi}_{q_1} 
+ \norm{\varphi}_{q_2}} \sqrt{\frac{N}{N-1}} \, \pbb{\alpha + \frac{1}{N}}\,.
\end{equation*}

\step{Conclusion of the proof.}
We have shown that the estimate \eqref{main estimate for alpha'} holds with
\begin{align*}
B_N(t) &\;=\; 2 \norm{w}_{L^{p_1} + L^{p_2}} \qB{\pb{\norm{\varphi(t)}_{r_1} + 
\norm{\varphi(t)}_{r_2}}^2 + 6 \pb{\norm{\varphi(t)}_{q_1} + \norm{\varphi(t)}_{q_2}}}\,,
\\
A_N(t) &\;=\; \frac{B_N(t)}{N}\,.
\end{align*}
Using $L^2$-norm conservation $\norm{\varphi(t)} = 1$ and interpolation we find 
$\norm{\varphi(t)}_{r_i}^2 \leq \norm{\varphi(t)}_{q_i}$. Thus,
\begin{equation*}
B_N(t) \;\leq\; 16 \norm{w}_{L^{p_1} + L^{p_2}} \pb{\norm{\varphi(t)}_{q_1} + 
\norm{\varphi(t)}_{q_2}}\,.
\end{equation*}
The claim now follows from the Gr\"onwall estimate \eqref{Gronwall}.

\section{Convergence for stronger singularities} \label{section: singular potentials}
In this section we extend the results of the Section \ref{L2 potentials} to more singular 
interaction potentials. We consider the case $w \in L^{p_0} + L^\infty$, where
\begin{equation} \label{definition of p_0}
\frac{1}{p_0} \;=\; \frac{1}{2} + \frac{1}{d}\,.
\end{equation}
For example in three dimensions $p_0 = 6/5$, which corresponds to singularities up to, but not 
including, the type $\abs{x}^{-5/2}$. Of course, there are other restrictions on the 
interaction potential which ensure the stability of the $N$-body Hamiltonian and the 
well-posedness of the Hartree equation. In practice, it is often these latter restrictions that 
determine the class of allowed singularities.

In the words of \cite{ReedSimonII} (p.\ 169), it is ``venerable physical folklore'' that an 
$N$-body Hamiltonian of the form \eqref{mean-field Hamiltonian}, with $h = -\Delta$ and $w(x) = 
\abs{x}^{-\zeta}$ for $\zeta < 2$, produces reasonable quantum dynamics in three dimensions.  
Mathematically, this means that such a Hamiltonian is self-adjoint; this is a well-known result 
(see e.g.\ \cite{ReedSimonII}). The corresponding Hartree equation is known to be globally 
well-posed (see \cite{GinibreVelo1980}). This section answers (affirmatively) the question 
whether, in the case of such singular interaction potentials, the mean-field limit of the 
$N$-body dynamics is governed by the Hartree equation.

\subsection{Outline and main result}
As in Section \ref{L2 potentials}, we need to control expressions of the form $\norm{w^2 * 
\abs{\varphi}^2}_\infty$. The situation is considerably more involved when $w^2$ is not locally 
integrable. An important step in dealing with such potentials in our proof is to express $w$ as 
the divergence of a vector field $\xi \in L^2$. This approach requires the control of not only 
$\alpha = \norm{q_1 \Psi}^2$ but also $\norm{\nabla_1 q_1 \Psi}^2$, which arises from 
integrating by parts in expressions containing the factor $\nabla \cdot \xi$. As it turns out, 
$\beta$, defined through
\begin{equation}
\beta_N(t) \;\deq\; \scalar{\Psi_N}{\widehat{n} \, \Psi_N} \big\vert_t\,,
\end{equation}
does the trick. This follows from an estimate exploiting conservation of energy (see Lemma 
\ref{lemma: energy estimate} below).
The inequality $m \leq n$ and the representation \eqref{two forms of alpha} yield
\begin{equation} \label{alpha < beta}
\alpha \;\leq\; \beta\,.
\end{equation}

We consider a Hamiltonian of the form \eqref{mean-field Hamiltonian} and make the following 
assumptions.
\begin{itemize}
\item[(B1)] The one-particle Hamiltonian $h$ is self-adjoint and bounded from below. Without 
loss of generality we assume that $h \geq 0$. We also assume that there are constants 
$\kappa_1, \kappa_2 > 0$ such that
\begin{equation*}
- \Delta \;\leq\; \kappa_1 \, h + \kappa_2\,,
\end{equation*}
as an inequality of forms on $\mathcal{H}^{(1)}$.
\item[(B2)]
The Hamiltonian \eqref{mean-field Hamiltonian} is self-adjoint and bounded from below. We also 
assume that $\mathcal{Q}(H_N) \subset X_N$, where $X_N$ is defined as in Assumption (A1).
\item[(B3)]
There is a constant $\kappa_3 \in (0,1)$ such that
\begin{equation*}
0 \;\leq\; (1 - \kappa_3) (h_1 + h_2) + W_{12}\,,
\end{equation*}
as an inequality of forms on $\mathcal{H}^{(2)}$.
\item[(B4)] The interaction potential $w$ is a real and even function satisfying $w \in L^{p} + 
L^\infty$, where $p_0 < p \leq 2$.
\item[(B5)] The solution $\varphi(\cdot)$ of \eqref{hartree equation} satisfies 
\begin{equation*}
\varphi(\cdot) \;\in\; C(\R; X_1^2 \cap L^\infty) \cap C^1(\R; L^2)\,,
\end{equation*}
where $X_1^2 \;\deq\; \mathcal{Q}(h^2) \subset L^2$ is equipped with the norm
\begin{equation*}
\norm{\varphi}_{X_1^2} \;\deq\; \normb{(1 + h^2)^{1/2} \varphi}\,.
\end{equation*}
\end{itemize}

Next, we define the microscopic energy per particle
\begin{equation*}
E^\Psi_N(t) \;\deq\; \frac{1}{N} \scalar{\Psi_N}{H_N \, \Psi_N} \big\vert_t\,,
\end{equation*}
as well as the Hartree energy
\begin{equation*}
E^\varphi(t) \;\deq\; \qbb{\scalar{\varphi}{h \, \varphi} + \frac{1}{2} \int \dd x \, \dd y \; 
w(x - y) \abs{\varphi(x)}^2 \abs{\varphi(y)}^2} \biggr\vert_t\,.
\end{equation*}
By spectral calculus, $E_N^\Psi(t)$ is independent of $t$. Also, invoking Assumption (B5) to 
differentiate $E^\varphi(t)$ with respect to $t$ shows that $E^\varphi(t)$ is conserved as 
well. Summarizing,
\begin{equation*}
E^\Psi_N(t) \;=\; E^\Psi_N(0)\,,\qquad E^\varphi(t) \;=\; E^\varphi(0)\,, \qquad t \in \R\,.
\end{equation*}
We may now state the main result of this section.
\begin{theorem} \label{theorem for singular potentials}
Let $\Psi_{N,0} \in \mathcal{Q}(H_N)$ and assume that Assumptions (B1) -- (B5) hold. Then there 
is a constant $K$, depending only on $d$, $h$, $w$ and $p$, such that
\begin{equation*}
\beta_N(t) \;\leq\; \pbb{\beta_N(0) + E^\Psi_N - E^\varphi + \frac{1}{N^\eta}} \, \ee^{K 
\phi(t)}\,,
\end{equation*}
where
\begin{equation} \label{definition of eta}
\eta \;\deq\; \frac{p/p_0 - 1}{2 p/p_0 - p/2 - 1}
\end{equation}
and
\begin{equation*}
\phi(t) \;\deq\; \int_0^t \dd s \; \pB{1 + \norm{\varphi(s)}^3_{X_1^2 \cap L^\infty}}\,.
\end{equation*}
\end{theorem}

\begin{remark}
We have convergence to the mean-field limit whenever $\lim_N E_N^\Psi = E^\varphi$ and $\lim_N 
\beta_N(0) = 0$. For instance if we start in a fully factorized state, $\Psi_{N,0} = 
\varphi_0^{\otimes N}$, then $\beta_N(0) = 0$ and
\begin{equation*}
E^\Psi_N - E^\varphi \;=\; \frac{1}{N} \scalar{\varphi_0 \otimes \varphi_0}{W_{12} \, \varphi_0 
\otimes \varphi_0}\,,
\end{equation*}
so that the Theorem \ref{theorem for singular potentials} yields
\begin{equation*}
E^{(1)}_N(t) \;\leq\; \beta_N(t) \;\lesssim\; \frac{1}{N^{\eta}} \ee^{K \phi(t)}\,,
\end{equation*}
and the analogue of Corollary \ref{corollary of main theorem} holds.
\end{remark}
\begin{remark}
The following graph shows the dependence of $\eta$ on $p$ for $d = 3$, i.e.\ $p_0 = 6/5$.
\vspace{0.3cm}
\begin{center}
\psfrag{p}[t][b]{$p$}
\psfrag{e}[][]{$\eta$}
\includegraphics[width=6cm]{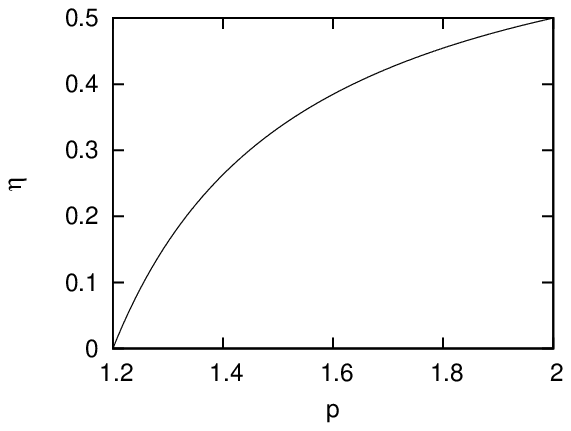}
\end{center}
\end{remark}
\begin{remark}
Theorem \ref{theorem for singular potentials} remains valid for a large class of time-dependent 
one-particle Hamiltonians $h(t)$. See Section \ref{section: remark on time-dependent 
potentials} below for a full discussion.
\end{remark}
\begin{remark}
In three dimensions Assumption (B1) and Sobolev's inequality imply that $\norm{\varphi}_\infty 
\lesssim \norm{\varphi}_{X_1^2}$, so that Assumption (B5) is equivalent to $\varphi \in 
C(\R; X_1^2) \cap C^1(\R; L^2)$.
\end{remark}

\subsection{Example: nonrelativistic particles with interaction potential of critical type}
Consider nonrelativistic particles in $\R^3$ with one-particle Hamiltonian $h = -\Delta$. The 
interaction potential is given by $w(x) = \lambda \abs{x}^{-2}$. This corresponds to a critical 
nonlinearity of the Hartree equation. We require that $\lambda > -1/2$, which ensures that the 
$N$-body Hamiltonian is stable and the Hartree equation has global solutions. To see this, 
recall Hardy's inequality in three dimensions,
\begin{equation} \label{Hardy in three dimensions}
\scalar{\varphi}{\abs{x}^{-2} \varphi} \;\leq\; 4 \scalar{\varphi}{-\Delta \varphi}\,.
\end{equation}
One easily infers that Assumptions (B1) -- (B3) hold. Moreover, Assumption (B4) holds for any 
$p < 3/2$.

In order to verify Assumption (B5) we refer to \cite{GinibreVelo1980}, where local 
well-posedness is proven. Global existence follows by standard methods using conservation of 
the mass $\norm{\varphi}^2$, conservation of the energy $E^\varphi$, and Hardy's inequality 
\eqref{Hardy in three dimensions}. Together they yield an a-priori bound on 
$\norm{\varphi}_{X_1}$, from which an a-priori bound for $\norm{\varphi}_{X_1^2}$ may be 
inferred; see \cite{GinibreVelo1980} for details.

We conclude: For any $\eta < 1/3$ there is a continuous function $\phi(t)$ such that Theorem 
\ref{theorem for singular potentials} holds.

\subsection{Proof of Theorem \ref{theorem for singular potentials}}
\subsubsection{An energy estimate}
In the first step of our proof we exploit conservation of energy to derive an estimate on 
$\norm{\nabla_1 q_1 \Psi}$.

\begin{lemma} \label{lemma: energy estimate}
Assume that Assumptions (B1) -- (B5) hold. Then \begin{equation*}
\norm{\nabla_1 q_1 \Psi}^2 \;\lesssim\; E^\Psi - E^\varphi + \pb{1 + \norm{\varphi}_{X_1^2 \cap 
L^\infty}^2} \pbb{\beta + \frac{1}{\sqrt{N}}}\,.
\end{equation*}
	\end{lemma}
\begin{proof}
Write
\begin{equation} \label{E^varphi written out}
E^\varphi \;=\; \scalar{\varphi}{h \varphi} + \frac{1}{2} \scalar{\varphi}{W^\varphi 
\varphi}\,,
\end{equation}
as well as
\begin{equation} \label{E^Psi written out}
E^\Psi \;=\; \scalar{\Psi}{h_1 \Psi} + \frac{N-1}{2N} \scalar{\Psi}{W_{12} \Psi}\,.
\end{equation}
Inserting
\begin{equation*}
\umat \;=\; p_1 p_2 + (\umat - p_1 p_2)
\end{equation*}
in front of every $\Psi$ in \eqref{E^Psi written out} and multiplying everything out yields
\begin{align*}
&\mspace{-60mu} \scalarb{\Psi}{(\umat - p_1 p_2) h_1 (\umat - p_1 p_2) \Psi} \\
\;=\; &\; E^\Psi - \scalar{\Psi}{p_1 p_2 h_1 p_1 p_2 \Psi}
\\
&{}-{} \frac{N-1}{2N} \scalar{\Psi}{p_1 p_2 W_{12} p_1 p_2 \Psi}
\\
&{}-{} \scalarb{\Psi}{(\umat - p_1 p_2) h_1 p_1 p_2 \Psi} - \scalarb{\Psi}{p_1 p_2 h_1 (\umat - 
p_1 p_2) \Psi}
\\
&{}-{} \frac{N-1}{2N} \scalarb{\Psi}{(\umat - p_1 p_2) W_{12} p_1 p_2 \Psi} - \frac{N-1}{2N} 
\scalarb{\Psi}{p_1 p_2 W_{12} (\umat - p_1 p_2) \Psi}
\\
&{}-{}\frac{N-1}{2N} \scalarb{\Psi}{(\umat - p_1 p_2) W_{12}(\umat - p_1 p_2) \Psi}\,.
\end{align*}
We want to find an upper bound for the left-hand side. In order to control the last term on the 
right-hand side for negative interaction potentials, we need to use some of the kinetic energy 
on the left-hand side. To this end, we split the left-hand side by multiplying it with $1 = 
\kappa_3 + (1 - \kappa_3)$. Thus, using \eqref{E^varphi written out}, we get
\begin{align}
&\mspace{-30mu} \kappa_3 \scalarb{\Psi}{(\umat - p_1 p_2) h_1 (\umat - p_1 p_2) \Psi} \notag \\
\;=\; &\; E^\Psi - E^\varphi
\notag \\
&{}-{} \scalar{\Psi}{p_1 p_2 h_1 p_1 p_2 \Psi} + \scalar{\varphi}{h \varphi}
\notag \\
&{}-{} \frac{N-1}{2N} \scalar{\Psi}{p_1 p_2 W_{12} p_1 p_2 \Psi} + \frac{1}{2} 
\scalar{\varphi}{W^\varphi \varphi}
\notag \\
&{}-{} \scalarb{\Psi}{(\umat - p_1 p_2) h_1 p_1 p_2 \Psi} - \scalarb{\Psi}{p_1 p_2 h_1 (\umat - 
p_1 p_2) \Psi}
\notag \\
&{}-{} \frac{N-1}{2N} \scalarb{\Psi}{(\umat - p_1 p_2) W_{12} p_1 p_2 \Psi} - \frac{N-1}{2N} 
\scalarb{\Psi}{p_1 p_2 W_{12} (\umat - p_1 p_2) \Psi}
\notag \\ \label{main formula in energy estimate}
&{}-{}\frac{N-1}{2N} \scalarb{\Psi}{(\umat - p_1 p_2) W_{12}(\umat - p_1 p_2) \Psi} - (1- 
\kappa_3) \scalarb{\Psi}{(\umat - p_1 p_2) h_1 (\umat - p_1 p_2) \Psi} \,.
\end{align}
The rest of the proof consists in estimating each line on the right-hand side of \eqref{main 
formula in energy estimate} separately. There is nothing to be done with the first line.

\step{Line 6.} The last line of \eqref{main formula in energy estimate} is equal to
\begin{multline*}
-\frac{N-1}{2N} \scalarB{\Psi}{(\umat - p_1 p_2) W_{12}(\umat - p_1 p_2) \Psi} - \frac{1}{2} 
(1- \kappa_3) \scalarb{\Psi}{(\umat - p_1 p_2) (h_1 + h_2) (\umat - p_1 p_2) \Psi} \\
\leq\; -\frac{N-1}{2N} \scalarB{\Psi}{(\umat - p_1p_2) \qb{(1 - \kappa_3)(h_1 + h_2) + W_{12}} 
(\umat - p_1 p_2) \Psi} \;\leq\; 0\,,
\end{multline*}
where in the last step we used Assumption (B3).

\step{Line 2.} The second line on the right-hand side of \eqref{main formula in energy 
estimate} is bounded in absolute value by
\begin{align*}
\absb{\scalar{\varphi}{h \varphi} - \scalar{\Psi}{p_1 p_2 h_1 p_1 p_2 \Psi}} &\;=\; 
\scalar{\varphi}{h \varphi} \absb{\scalar{\Psi}{(\umat - p_1 p_2)\Psi}}
\\
&\;=\; \scalar{\varphi}{h \varphi} \absb{\scalar{\Psi}{(q_1 p_2 + p_1 q_2 + q_1 q_2)\Psi}} \\
&\;\leq\; 3 \, \alpha \, \scalar{\varphi}{h \varphi}
\\
&\;\leq\; 3 \, \beta \, \scalar{\varphi}{h \varphi}\,,
\end{align*}
where in the last step we used \eqref{alpha < beta}.

\step{Line 3.} The third line on the right-hand side of \eqref{main formula in energy estimate} 
is bounded in absolute value by
\begin{align*}
\absbb{\frac{1}{2}\scalar{\varphi}{W^\varphi \varphi} - \frac{N-1}{2N} \scalar{\Psi}{p_1 p_2 
W_{12} p_1 p_2 \Psi}}
&\;=\; \frac{1}{2}\absb{\scalar{\varphi}{W^\varphi \varphi}}\, \absbb{1 - \frac{N-1}{N} 
\scalar{\Psi}{p_1 p_2 \Psi}}
\\
&\;\leq\; \frac{1}{2}\norm{W^\varphi}_\infty \absbb{\scalarb{\Psi}{(q_1 p_2 + p_1 q_2 + q_1 
q_2) \Psi} + \frac{1}{N} \scalar{\Psi}{p_1 p_2 \Psi}}
\\
&\;\leq\; \frac{3}{2} \norm{W^\varphi}_\infty \pbb{\alpha + \frac{1}{N}}
\\
&\;\leq \frac{3}{2} \norm{W^\varphi}_\infty \pbb{\beta + \frac{1}{N}}\,.
\end{align*}
As in \eqref{bound of convolution with w}, one finds that
\begin{equation*}
\norm{W^\varphi}_\infty \;\leq\; \norm{w}_{L^1 + L^\infty} \norm{\varphi}_{L^2 \cap 
L^\infty}^2\,.
\end{equation*}

\step{Line 4.} The fourth line on the right-hand side of \eqref{main formula in energy 
estimate} is bounded in absolute value by
\begin{align*}
\absb{\scalarb{\Psi}{(\umat - p_1 p_2) h_1 p_1 p_2 \Psi}} &\;=\; \absb{\scalarb{\Psi}{(q_1 p_2 
+ p_1 q_2 + q_1 q_2) h_1 p_1 p_2 \Psi}}
\\
&\;=\; \absb{\scalarb{\Psi}{q_1 h_1 p_1 p_2 \Psi}}
\\
&\;=\; \absb{\scalarb{\Psi}{q_1 \, \widehat{n}^{-1/2} \, \widehat{n}^{1/2} \, h_1 p_1 p_2 
\Psi}}
\\
&\;=\; \absb{\scalarb{\Psi}{q_1 \, \widehat{n}^{-1/2} \, h_1 \, \widehat{\tau_1 n}^{1/2} \, p_1 
p_2 \Psi}}\,,
\end{align*}
where in the last step we used Lemma \ref{pulling projectors through}. Using Cauchy-Schwarz, we 
thus get
\begin{align*}
\absb{\scalarb{\Psi}{(\umat - p_1 p_2) h_1 p_1 p_2 \Psi}} &\;\leq\; \sqrt{\scalarb{\Psi}{q_1 
\widehat{n}^{-1} \Psi}} \, \sqrt{\scalarb{\Psi}{p_1 p_2 \,\widehat{\tau_1 n}^{1/2}\, h_1^2 \, 
\widehat{\tau_1 n}^{1/2} \, p_1 p_2 \Psi}}
\\
&\;=\; \sqrt{\scalar{\Psi}{\widehat{n} \,\Psi}} \sqrt{\scalar{\varphi}{h^2 \varphi}} 
\sqrt{\scalarb{\Psi}{\widehat{\tau_1 n} \, p_1 p_2 \Psi}}\,,
\end{align*}
where in the second step we used Lemma \ref{replacing q's with n's}. Using
\begin{equation*}
(\tau_1 n)(k) \;=\; \sqrt{\frac{k + 1}{N}} \;\leq\; n(k) + \frac{1}{\sqrt{N}}
\end{equation*}
we find
\begin{align*}
\absb{\scalarb{\Psi}{(\umat - p_1 p_2) h_1 p_1 p_2 \Psi}} &\;\leq\; \sqrt{\beta} 
\sqrt{\scalar{\varphi}{h^2 \varphi}} \sqrt{\scalar{\Psi}{\widehat{n} \, \Psi} + 
\frac{1}{\sqrt{N}}}
\\
&\;=\; \sqrt{\scalar{\varphi}{h^2 \varphi}} \sqrt{\beta} \pbb{\sqrt{\beta} + \frac{1}{N^{1/4}}}
\\
&\;\leq\; 2 \sqrt{\scalar{\varphi}{h^2 \varphi}} \pbb{\beta + \frac{1}{\sqrt{N}}}\,.
\end{align*}

\step{Line 5.} Finally, we turn our attention to the fifth line on the right-hand side of 
\eqref{main formula in energy estimate}, which is bounded in absolute value by
\begin{equation*}
\absb{\scalarb{\Psi}{p_1 p_2 W_{12} (\umat - p_1 p_2) \Psi}} \;=\; \absb{\scalarb{\Psi}{p_1 p_2 
W_{12} (p_1 q_2 + q_1 p_2 + q_1 q_2\Psi}} \;\leq\; 2 \mathrm{(a)} + \mathrm{(b)}\,,
\end{equation*}
where
\begin{equation*}
\mathrm{(a)} \;\deq\; \absb{\scalarb{\Psi}{p_1 p_2 W_{12} q_1 p_2 \Psi}}\,, \qquad \mathrm{(b)} 
\;\deq\; \absb{\scalarb{\Psi}{p_1 p_2 W_{12} q_1 q_2 \Psi}}\,.
\end{equation*}
One finds, using \eqref{W sandwiched between p's}, Lemma \ref{pulling projectors through} and 
Lemma \ref{replacing q's with n's},
\begin{align*}
\mathrm{(a)} &\;=\; \absb{\scalarb{\Psi}{p_1 p_2 W_1^\varphi q_1 \Psi}}
\\
&\;=\; \absb{\scalarb{\Psi}{p_1 p_2 W_1^\varphi \, \widehat{n}^{1/2} \, \widehat{n}^{-1/2} \, 
q_1 \Psi}}
\\
&\;=\; \absb{\scalarb{\Psi}{p_1 p_2 \, \widehat{\tau_1 n}^{1/2} \, W_1^\varphi \, 
\widehat{n}^{-1/2} \, q_1 \Psi}}
\\
&\;\leq\; \norm{W^\varphi}_\infty \, \sqrt{\scalarb{\Psi}{\widehat{\tau_1 n} \, \Psi}} \, 
\sqrt{\scalarb{\Psi}{\widehat{n}^{-1} \, q_1 \Psi}}
\\ &\;\leq\; \norm{W^\varphi}_\infty \, \sqrt{\scalarb{\Psi}{\widehat{n} \, \Psi} + 
\frac{1}{\sqrt{N}}} \, \sqrt{\scalarb{\Psi}{\widehat{n}\,  \Psi}}
\\
&\;\leq\; 2 \norm{W^\varphi}_\infty \pbb{\beta + \frac{1}{\sqrt{N}}}\,.
\end{align*}

The estimation of (b) requires a little more effort. We start by splitting
\begin{equation*}
w \;=\; w^{(p)} + w^{(\infty)}\,, \qquad w^{(p)} \in L^p \,,\; w^{(\infty)} \in L^\infty\,.
\end{equation*}
This yields $\mathrm{(b)} \leq \mathrm{(b)}^{(p)} + \mathrm{(b)}^{(\infty)}$ in 
self-explanatory notation. Let us first concentrate on $\mathrm{(b)}^{(\infty)}$:
\begin{align*}
\mathrm{(b)}^{(\infty)} &\;=\; \absb{\scalarb{\Psi}{p_1 p_2 W_{12}^{(\infty)} q_1 q_2 \Psi}}
\\
&\;=\; \absb{\scalarb{\Psi}{p_1 p_2 W_{12}^{(\infty)} \, \widehat{n} \, \widehat{n}^{-1} \, q_1 
q_2 \Psi}}
\\
&\;=\; \absb{\scalarb{\Psi}{p_1 p_2 \, \widehat{\tau_2 n} \, W_{12}^{(\infty)} \, 
\widehat{n}^{-1} \, q_1 q_2 \Psi}}
\\
&\;\leq\; \norm{W^{(\infty)}}_\infty \sqrt{\scalarb{\Psi}{\widehat{\tau_2 n}^2 \Psi}} \, 
\sqrt{\scalarb{\Psi}{\widehat{n}^{-2} \, q_1 q_2 \Psi}}
\\
&\;\leq\; \norm{w^{(\infty)}}_\infty \, \sqrt{\alpha + \frac{2}{N}} \, \sqrt{\alpha}
\\
&\;\leq\; 2 \norm{w^{(\infty)}}_\infty \, \pbb{\beta + \frac{2}{N}}\,.
\end{align*}

Let us now consider $\mathrm{(b)}^{(p)}$. In order to deal with the singularities in $w^{(p)}$, 
we write it as the divergence of a vector field $\xi$,
\begin{equation} \label{w as divergence of xi}
w^{(p)} \;=\; \nabla \cdot \xi\,.
\end{equation}
This is nothing but a problem of electrostatics, which is solved by
\begin{equation*}
\xi \;=\; C \, \frac{x}{\abs{x}^d} * w^{(p)}\,,
\end{equation*}
with some constant $C$ depending on $d$. By the Hardy-Littlewood-Sobolev inequality, we find
\begin{equation} \label{bound on xi}
\norm{\xi}_q \;\lesssim\; \normb{w^{(p)}}_p\,, \qquad \frac{1}{q} \;=\; \frac{1}{p} - 
\frac{1}{d}\,.
\end{equation}
Thus if $p \geq p_0$ then $q \geq 2$. Denote by $X_{12}$ multiplication by $\xi(x_1 - x_2)$.  
For the following it is convenient to write $\nabla \cdot \xi = \nabla^\rho \xi^\rho$, where a 
summation over $\rho = 1, \dots, d$ is implied.

Recalling Lemma \ref{pulling projectors through}, we therefore get
\begin{align*}
\mathrm{(b)}^{(p)} &\;=\; \absb{\scalarb{\Psi}{p_1 p_2 W_{12}^{(p)} \, \widehat{n} \, 
\widehat{n}^{-1} q_1 q_2 \Psi}}
\\
&\;=\; \absb{\scalarb{\Psi}{p_1 p_2\, \widehat{\tau_2 n}\, W_{12}^{(p)} \, \widehat{n}^{-1} q_1 
q_2 \Psi}} \\
&\;=\; \absb{\scalarb{\Psi}{p_1 p_2\, \widehat{\tau_2 n}\, (\nabla_1^\rho X^\rho)_{12} \, 
\widehat{n}^{-1} q_1 q_2 \Psi}}\,.
\end{align*}
Integrating by parts yields
\begin{equation} \label{b^p in energy estimate}
\mathrm{(b)}^{(p)} \;\leq\; \absb{\scalarb{\nabla_1^\rho \widehat{\tau_2 n}\,p_1 p_2 
\Psi}{X^\rho_{12} \, \widehat{n}^{-1} q_1 q_2 \Psi}} + \absb{\scalarb{\widehat{\tau_2 n}\,p_1 
p_2 \Psi}{X^\rho_{12} \nabla_1^\rho\, \widehat{n}^{-1} q_1 q_2 \Psi}}\,.
\end{equation}
Let us begin by estimating the first term. Recalling that $p = \ket{\varphi}\bra{\varphi}$, we 
find that the first term on the right-hand side of \eqref{b^p in energy estimate} is equal to
\begin{align*}
\absb{\scalarb{X^\rho_{12} \, p_2 (\nabla^\rho p)_1 \widehat{\tau_2 n}\, \Psi}{\widehat{n}^{-1} 
q_1 q_2 \Psi}} &\;\leq\; \sqrt{\scalarb{(\nabla^\rho p)_1 \widehat{\tau_2 n}\, \Psi}{p_2 
X^\rho_{12} X^\sigma_{12} \, p_2 (\nabla^\sigma p)_1 \widehat{\tau_2 n}\, \Psi}} \, 
\normb{\widehat{n}^{-1} q_1 q_2 \Psi}
\\
&\;\leq\; \sqrt{\normb{\abs{\varphi}^2 * \xi^2}_\infty} \, \norm{\nabla \varphi} \, 
\normb{\widehat{\tau_2 n}\, \Psi}\, \normb{\widehat{n}^{-1} q_1 q_2 \Psi}
\\
&\;\lesssim\; \norm{\xi}_q \, \norm{\varphi}_{L^2 \cap L^\infty} \, \norm{\varphi}_{X_1} 
\sqrt{\alpha + \frac{2}{N}} \, \sqrt{\alpha}\,,
\end{align*}
where we used Young's inequality, Assumption (B1), and Lemma \ref{replacing q's with n's}. 
Recalling that $\beta \leq \alpha$, we conclude that the first term on the right-hand side of 
\eqref{b^p in energy estimate} is bounded by
\begin{equation*}
C \, \norm{\varphi}_{X_1\cap L^\infty}^2 \pbb{\beta + \frac{1}{N}}\,.
\end{equation*}

Next, we estimate the second term on the right-hand side of \eqref{b^p in energy estimate}. It 
is equal to
\begin{align*}
\absb{\scalarb{X_{12}^\rho\, p_1 p_2 \,\widehat{\tau_2 n}\, \Psi}{\nabla_1^\rho\, 
\widehat{n}^{-1} q_1 q_2 \Psi}} &\;\leq\; \sqrt{\scalarb{\widehat{\tau_2 n} \, \Psi}{ p_1 p_2 
X_{12}^2 p_1 p_2 \, \widehat{\tau_2 n} \, \Psi}} \, \normb{\nabla_1 \, \widehat{n}^{-1} \, q_1 
q_2 \Psi}
\\
&\;\leq\; \sqrt{\normb{\abs{\varphi}^2 * \xi^2}_\infty} \, \norm{\widehat{\tau_2 n} \, \Psi}\, 
\normb{\nabla_1 \, \widehat{n}^{-1} \, q_1 q_2 \Psi}
\\
&\;\leq\; \norm{\xi}_q \, \norm{\varphi}_{L^2 \cap L^\infty} \sqrt{\alpha + \frac{2}{N}} \, 
\normb{\nabla_1 \, \widehat{n}^{-1} \, q_1 q_2 \Psi}\,.
\end{align*}
We estimate $\normb{\nabla_1 \, \widehat{n}^{-1} \, q_1 q_2 \Psi}$ by introducing $\umat = p_1 
+ q_1$ on the left. The term arising from $p_1$ is bounded by
\begin{align*}
\normb{p_1 \nabla_1 \, \widehat{n}^{-1} \, q_1 q_2 \Psi} &\;=\; \normb{p_1 q_2 \, 
\widehat{\tau_1 n}^{-1}\, \nabla_1  q_1 \Psi}
\\
&\;\leq\; \sqrt{\scalarb{\nabla_1  q_1 \Psi}{q_2 \, \widehat{\tau_1 n}^{-2}\, \nabla_1  q_1 
\Psi}}
\\
&\;=\; \sqrt{\scalarbb{\nabla_1  q_1 \Psi}{\frac{1}{N-1} \sum_{i = 2}^N q_i \, \widehat{\tau_1 
n}^{-2}\, \nabla_1  q_1 \Psi}}
\\
&\;\leq\; \sqrt{\scalarbb{\nabla_1  q_1 \Psi}{\frac{1}{N} \sum_{i = 1}^N q_i \, \widehat{\tau_1 
n}^{-2}\, \nabla_1  q_1 \Psi}}
\\
&\;=\; \sqrt{\scalarb{\nabla_1  q_1 \Psi}{\widehat{n}^2 \, \widehat{\tau_1 n}^{-2}\, \nabla_1  
q_1 \Psi}}
\\
&\;\leq\; \norm{\nabla_1 q_1 \Psi}\,.
\end{align*}
The term arising from $q_1$ in the above splitting is dealt with in exactly the same way. Thus 
we have proven that the second term on the right-hand side of \eqref{b^p in energy estimate} is 
bounded by
\begin{equation*}
C \norm{\varphi}_{L^2 \cap L^\infty} \sqrt{\beta + \frac{1}{N}} \, \norm{\nabla_1 q_1 \Psi}\,.
\end{equation*}

Summarizing, we have
\begin{equation*}
\mathrm{(b)}^{(p)} \;\lesssim\; \norm{\varphi}_{X_1 \cap L^\infty}^2 \pbb{\beta + \frac{1}{N}} 
+ \norm{\varphi}_{L^2 \cap L^\infty} \sqrt{\beta + \frac{1}{N}} \, \norm{\nabla_1 q_1 \Psi}\,.
\end{equation*}

\step{Conclusion of the proof.} Putting all the estimates of the right-hand side of \eqref{main 
formula in energy estimate} together, we find
\begin{multline}
\scalarb{\Psi}{(\umat - p_1 p_2) h_1 (\umat - p_1 p_2) \Psi}  \\
\lesssim\; E^\Psi - E^\varphi + \pb{1 + \norm{\varphi}_{X_1^2 \cap L^\infty}^2} \pbb{\beta + 
\frac{1}{\sqrt{N}}}
+ \norm{\varphi}_{L^2 \cap L^\infty} \sqrt{\beta + \frac{1}{N}} \, \norm{\nabla_1 q_1 \Psi}\,.
\label{energy bound with two p's}
\end{multline}
Next, from $\umat - p_1 p_2 = p_1 q_2 + q_1$
we deduce
\begin{equation*}
\norm{\sqrt{h_1} q_1 \Psi} \;=\; \normb{\sqrt{h_1} (\umat - p_1 p_2) \Psi - \sqrt{h_1} p_1 q_2 
\Psi}
\;\leq\; \normb{\sqrt{h_1} (\umat - p_1 p_2) \Psi} + \norm{\sqrt{h_1} p_1 q_2 \Psi}\,.
\end{equation*}
Now, recalling that $p = \ket{\varphi} \bra{\varphi}$, we find
\begin{equation*}
\norm{\sqrt{h_1} p_1 q_2 \Psi} \;\leq\; \norm{\sqrt{h_1} p_1} \norm{q_2 \Psi} \;\leq\; 
\norm{\varphi}_{X_1} \sqrt{\beta}\,.
\end{equation*}
Therefore,
\begin{equation*}
\norm{\sqrt{h_1} q_1 \Psi}^2 \;\lesssim\; \normb{\sqrt{h_1} (\umat - p_1 p_2) \Psi}^2 + 
\norm{\varphi}_{X_1}^2\beta\,.
\end{equation*}
Plugging in \eqref{energy bound with two p's} yields
\begin{equation*}
\normb{\sqrt{h_1} q_1 \Psi}^2 \;\lesssim\; E^\Psi - E^\varphi + \pb{1 + \norm{\varphi}_{X_1^2 
\cap L^\infty}^2} \pbb{\beta + \frac{1}{\sqrt{N}}}
+ \norm{\varphi}_{L^2 \cap L^\infty} \sqrt{\beta + \frac{1}{N}} \, \norm{\nabla_1 q_1 \Psi}\,.
\end{equation*}
Next, we observe that Assumption (B1) implies
\begin{equation*}
\norm{\nabla_1 q_1 \Psi} \;\lesssim\; \normb{\sqrt{h_1} q_1 \Psi} + \sqrt{\beta}\,,
\end{equation*}
so that we get
\begin{equation*}
\normb{\sqrt{h_1} q_1 \Psi}^2 \;\lesssim\; E^\Psi - E^\varphi + \pb{1 + \norm{\varphi}_{X_1^2 
\cap L^\infty}^2} \pbb{\beta + \frac{1}{\sqrt{N}}}
+ \norm{\varphi}_{L^2 \cap L^\infty} \sqrt{\beta + \frac{1}{N}} \, \norm{\sqrt{h_1} q_1 
\Psi}\,.
\end{equation*}
Now we claim that
\begin{equation} \label{energy estimate for h}
\normb{\sqrt{h_1} q_1 \Psi}^2 \;\lesssim\; E^\Psi - E^\varphi + \pb{1 + \norm{\varphi}_{X_1^2 
\cap L^\infty}^2} \pbb{\beta + \frac{1}{\sqrt{N}}}\,.
\end{equation}
This follows from the general estimate
\begin{equation*}
x^2 \;\leq\; C (R + a x) \quad \Longrightarrow \quad x^2 \;\leq\; 2 CR + C^2 a^2\,,
\end{equation*}
which itself follows from the elementary inequality
\begin{equation*}
C (R + a x) \;\leq\; CR + \frac{1}{2} C^2 a^2 + \frac{1}{2} x^2\,.
\end{equation*}
The claim of the Lemma now follows from \eqref{energy estimate for h} by using Assumption (B1).
\end{proof}

\subsubsection{A bound on $\dot{\beta}$}
We start exactly as in Section \ref{L2 potentials}. Assumptions (B1) -- (B5) imply that $\beta$ 
is differentiable in $t$ with derivative
\begin{align}
\dot{\beta} &\;=\; \frac{\ii}{2} \scalarb{\Psi}{\comb{(N-1) W_{12} - N W^\varphi_1 - N 
W^\varphi_2}{\widehat{n}}\Psi}
\notag \\ &\;=\; 2 \mathrm{(I)} + 2 \mathrm{(II)} + \mathrm{(III)} + \text{ complex 
conjugate}\,,
\end{align}
where
\begin{align*}
\mathrm{(I)} &\;\deq\; \frac{\ii}{2} \scalarb{\Psi}{p_1 p_2 \comb{(N-1) W_{12} - N W^\varphi_1 
- N W^\varphi_2}{\widehat{n}} q_1 p_2\Psi}\,,
\\
\mathrm{(II)} &\;\deq\; \frac{\ii}{2} \scalarb{\Psi}{q_1 p_2 \comb{(N-1) W_{12} - N W^\varphi_1 
- N W^\varphi_2}{\widehat{n}} q_1 q_2\Psi}\,,
\\
\mathrm{(III)} &\;\deq\; \frac{\ii}{2} \scalarb{\Psi}{p_1 p_2 \comb{(N-1) W_{12} - N 
W^\varphi_1 - N W^\varphi_2}{\widehat{n}} q_1 q_2\Psi}\,.
\end{align*}

\step{Term $\mathrm{(I)}$.} Using \eqref{W sandwiched between p's} we find
\begin{align*}
2 \absb{\mathrm{(I)}} &\;=\; \absb{\scalarb{\Psi}{p_1 p_2 \comb{(N-1) W_{12} - N W^\varphi_1 - 
N W^\varphi_2}{\widehat{n}} q_1 p_2\Psi}}
\\
&\;=\; \absb{\scalarb{\Psi}{p_1 p_2 \comb{W^\varphi_1}{\widehat{n}} q_1 p_2\Psi}}
\\
&\;=\; \absb{\scalarb{\Psi}{p_1 p_2 W^\varphi_1 \pb{\widehat{n} - \widehat{\tau_{-1} n}} q_1 
p_2\Psi}}\,,
\end{align*}
where we used Lemma \ref{pulling projectors through}. Define
\begin{equation} \label{derivative of n}
\mu(k) \;\deq\; N\pb{n(k) - (\tau_{-1}n)(k)} \;=\;  \frac{\sqrt{N}}{\sqrt{k} + \sqrt{k - 1}} 
\;\leq\; n^{-1}(k)\,, \qquad k = 1, \dots, N\,.
\end{equation}
Thus,
\begin{align*}
\absb{\mathrm{(I)}} &\;=\; \frac{1}{N}\absb{\scalarb{\Psi}{p_1 p_2 W^\varphi_1 
\,\widehat{\mu}\, q_1 p_2\Psi}}
\\
&\;\leq\; \frac{1}{N}\norm{W^\varphi}_\infty \sqrt{\scalarb{\Psi}{\widehat{\mu}^2 \, q_1 \Psi}}
\\
&\;\leq\; \frac{1}{N} \norm{W^\varphi}_\infty \sqrt{\scalarb{\Psi}{\widehat{n}^{-2} \, q_1 
\Psi}}
\\
&\;\lesssim\; \frac{1}{N} \norm{\varphi}^2_{L^2 \cap L^\infty}\,,
\end{align*}
by \eqref{replacing q with m 1}.

\step{Term $\mathrm{(II)}$.} Using Lemma \ref{pulling projectors through} we find
\begin{align}
2 \abs{\mathrm{(II)}} &\;=\; \absb{\scalarb{\Psi}{q_1 p_2 \comb{(N-1) W_{12} - N 
W^\varphi_2}{\widehat{n}} q_1 q_2\Psi}}
\\
&\;=\; \absbb{\scalarbb{\Psi}{q_1 p_2 \pbb{\frac{N-1}{N} W_{12} - W^\varphi_2} \, \widehat{\mu} 
\, q_1 q_2\Psi}}
\\ \label{term II for beta}
&\;\leq\; \underbrace{\absb{\scalarb{\Psi}{q_1 p_2 W_{12} \,\widehat{\mu}\, q_1 
q_2\Psi}}}_{\eqd \; \mathrm{(a)}} + \underbrace{\absb{\scalarb{\Psi}{q_1 p_2 W^\varphi_2 
\,\widehat{\mu}\, q_1 q_2\Psi}}}_{\eqd \; \mathrm{(b)}}\,.
\end{align}
One immediately finds
\begin{equation*}
\mathrm{(b)} \;\leq\;\norm{W^\varphi}_\infty \, \norm{q_1 \Psi} 
\sqrt{\scalar{\Psi}{\widehat{\mu}^2 q_1 q_2 \Psi}} \;\lesssim\; \norm{\varphi}^2_{L^2 \cap 
L^\infty} \beta\,.
\end{equation*}
In (a) we split
\begin{equation*}
w \;=\; w^{(p)} + w^{(\infty)}\,, \qquad w^{(p)} \in L^p \,,\; w^{(\infty)} \in L^\infty\,,
\end{equation*}
with a resulting splitting $\mathrm{(a)} \leq \mathrm{(a)}^{(p)} + \mathrm{(a)}^{(\infty)}$. 
The easy part is
\begin{align*}
\mathrm{(a)}^{(\infty)} \;\leq\; \norm{w^{(\infty)}}_\infty \, \norm{q_1 \Psi}^2 \;\lesssim\; 
\beta\,.
\end{align*}
In order to deal with $\mathrm{(a)}^{(p)}$ we write $w^{(p)} = \nabla \cdot \xi$ as the 
divergence of a vector field $\xi$, exactly as in the proof of Lemma \ref{lemma: energy 
estimate}; see \eqref{w as divergence of xi} and the remarks after it. We integrate by parts to 
find
\begin{align}
\mathrm{(a)}^{(p)} &\;=\; \absb{\scalarb{\Psi}{q_1 p_2 (\nabla_1^\rho X^\rho)_{12} 
\,\widehat{\mu}\, q_1 q_2\Psi}}
\notag \\ \label{a^p split in two}
&\;\leq\; \absb{\scalarb{\nabla_1^\rho q_1 p_2 \Psi}{X_{12}^\rho \, \widehat{\mu} \, q_1 q_2 
\Psi}} + \absb{\scalarb{q_1 p_2 \Psi}{X_{12}^\rho \nabla_1^\rho \, \widehat{\mu} \, q_1 q_2 
\Psi}}\,.
\end{align}
The first term of \eqref{a^p split in two} is equal to
\begin{align*}
\absb{\scalarb{X_{12}^\rho p_2 \nabla_1^\rho q_1 \Psi}{\widehat{\mu} \, q_1 q_2 \Psi}} 
&\;\leq\; \sqrt{\scalarb{\nabla_1^\rho q_1 \Psi}{p_2 X_{12}^\rho X_{12}^\sigma p_2 
\nabla_1^\sigma q_1 \Psi}} \, \sqrt{\scalarb{\Psi}{\widehat{\mu}^2\, q_1 q_2 \Psi}}
\\
&\;\lesssim\; \sqrt{\norm{\xi^2 * \abs{\varphi}^2}_\infty} \, \norm{\nabla_1 q_1 \Psi} \, 
\sqrt{\scalarb{\Psi}{\widehat{n}^{-2}\, q_1 q_2 \Psi}}
\\
&\;\leq\; \sqrt{\norm{\xi^2 * \abs{\varphi}^2}_\infty} \, \norm{\nabla_1 q_1 \Psi} \, 
\sqrt{\frac{N}{N-1}\scalarb{\Psi}{\widehat{n}^2 \Psi}}
\\
&\;\lesssim\; \norm{\xi}_q \, \norm{\varphi}_{L^2 \cap L^\infty} \, \norm{\nabla_1 q_1 \Psi} \, 
\sqrt{\beta}
\\
&\;\lesssim\; \norm{\nabla_1 q_1 \Psi}^2\, \norm{\varphi}_{L^2 \cap L^\infty} + \beta\, 
\norm{\varphi}_{L^2 \cap L^\infty}\,,
\end{align*}
where in the second step we used \eqref{derivative of n}, in the third Lemma \ref{replacing q's 
with n's}, and in the last \eqref{alpha < beta}, Young's inequality, and \eqref{bound on xi}.
The second term of \eqref{a^p split in two} is equal to
\begin{multline} \absb{\scalarb{q_1 p_2 \Psi}{X_{12}^\rho (p_1 + q_1) \nabla_1^\rho \, 
\widehat{\mu} \, q_1 q_2 \Psi}}
\\ \label{second term of a^p split in two}
\;\leq\;
\absb{\scalarb{q_1 p_2 \Psi}{X_{12}^\rho p_1 \, \widehat{\tau_1 \mu} \, \nabla_1^\rho q_1 q_2 
\Psi}} + \absb{\scalarb{q_1 p_2 \Psi}{X_{12}^\rho q_1 \, \widehat{\mu} \, \nabla_1^\rho  q_1 
q_2 \Psi}}\,,
\end{multline}
where we used Lemma \ref{pulling projectors through}. We estimate the first term of 
\eqref{second term of a^p split in two}. The second term is dealt with in exactly the same way. 
We find
\begin{align*}
\absb{\scalarb{p_1 X_{12}^\rho q_1 p_2 \Psi}{\widehat{\tau_1 \mu} \, \nabla_1^\rho q_1 q_2 
\Psi}} &\;\leq\; \sqrt{\scalarb{\Psi}{q_1 p_2 X_{12}^2 p_2 q_1 \Psi}} \, 
\sqrt{\scalarb{\nabla_1 q_1 \Psi}{q_2 \, \widehat{\tau_1 \mu}^2 \, q_2 \nabla_1 q_1 \Psi}}
\\
&\;\leq\; \sqrt{\norm{\xi^2 * \abs{\varphi}^2}_\infty} \, \norm{q_1 \Psi} \, 
\sqrt{\scalarb{\nabla_1 q_1 \Psi}{\widehat{n}^{-2} \, q_2 \nabla_1 q_1 \Psi}}
\\
&\;\lesssim \norm{\xi}_q \,\norm{\varphi}_{L^2 \cap L^\infty}\, \sqrt{\alpha} \, 
\sqrt{\frac{1}{N-1} \sum_{i = 2}^N \scalarb{\nabla_1 q_1 \Psi}{\widehat{n}^{-2} \, q_i \nabla_1 
q_1 \Psi}}
\\
&\;\lesssim\; \norm{\varphi}_{L^2 \cap L^\infty} \sqrt{\beta} \sqrt{\frac{1}{N-1} \sum_{i = 
1}^N \scalarb{\nabla_1 q_1 \Psi}{\widehat{n}^{-2} \, q_i \nabla_1 q_1 \Psi}}
\\
&\;=\; \norm{\varphi}_{L^2 \cap L^\infty} \sqrt{\beta} \sqrt{\frac{N}{N-1} \scalarb{\nabla_1 
q_1 \Psi}{\widehat{n}^{-2} \, \widehat{n}^2 \, \nabla_1 q_1 \Psi}}
\\
&\;\lesssim\; \norm{\varphi}_{L^2 \cap L^\infty} \sqrt{\beta} \, \norm{\nabla_1 q_1 \Psi}
\\
&\;\leq\; \beta \, \norm{\varphi}_{L^2 \cap L^\infty} + \norm{\nabla_1 q_1 \Psi}^2 \, 
\norm{\varphi}_{L^2 \cap L^\infty}\,.
\end{align*}
In summary, we have proven that
\begin{equation*}
\absb{\mathrm{(II)}} \;\lesssim\; \beta \,\norm{\varphi}_{L^2 \cap L^\infty}+ \norm{\nabla_1 
q_1 \Psi}^2\,\norm{\varphi}_{L^2 \cap L^\infty}\,.
\end{equation*}

\step{Term $\mathrm{(III)}$.} Using Lemma \ref{pulling projectors through} we find
\begin{equation*}
2 \abs{\mathrm{(III)}} \;=\; (N-1) \absb{ \scalarb{\Psi}{p_1 p_2 \comb{W_{12} }{\widehat{n}} 
q_1 q_2\Psi}} \;=\; (N-1) \absb{ \scalarb{\Psi}{p_1 p_2 W_{12} \pb{\widehat{n} - 
\widehat{\tau_{-2} n}} q_1 q_2\Psi}} \,.
\end{equation*}
Defining
\begin{equation}
\nu(k) \;\deq\; N\pb{n(k) - (\tau_{-2}n)(k)} \;=\; \frac{\sqrt{N}}{\sqrt{k} + \sqrt{k - 2}} 
\;\leq\; n^{-1}(k)\,, \qquad k = 2, \dots, N\,,
\end{equation}
we have
\begin{equation*}
2 \absb{\mathrm{(III)}} \;\leq\; \absb{ \scalarb{\Psi}{p_1 p_2 W_{12} \, \widehat{\nu}\, q_1 
q_2\Psi}}
\end{equation*}
As usual we start by splitting \begin{equation*}
w \;=\; w^{(p)} + w^{(\infty)}\,, \qquad w^{(p)} \in L^p \,,\; w^{(\infty)} \in L^\infty\,,
\end{equation*}
with the induced splitting $\mathrm{(III)} = \mathrm{(III)}^{(p)} + \mathrm{(III)}^{(\infty)}$. 
Thus, using Lemma \ref{pulling projectors through}, we find
\begin{align*}
2 \absb{\mathrm{(III)}^{(\infty)}} &\;=\; \absb{ \scalarb{\Psi}{p_1 p_2 W^{(\infty)}_{12} \, 
\widehat{n}^{1/2} \, \widehat{n}^{-1/2} \, \widehat{\nu}\, q_1 q_2\Psi}}
\\
&\;=\; \absb{ \scalarb{\Psi}{p_1 p_2 \, \widehat{\tau_2 n}^{1/2} \, W^{(\infty)}_{12} \, 
\widehat{n}^{-1/2} \, \widehat{\nu}\, q_1 q_2\Psi}}
\\
&\;\leq\; \norm{w^{(\infty)}}_\infty \sqrt{\scalarb{\Psi}{\widehat{\tau_2 n} \, \Psi}} \, 
\sqrt{\scalarb{\Psi}{\widehat{n}^{-1} \, \widehat{\nu}^2 \, q_1 q_2 \Psi}}
\\
&\;\lesssim \sqrt{\beta + \sqrt{\frac{2}{N}}} \, \sqrt{\scalarb{\Psi}{\widehat{n}^{-3} \, q_1 
q_2 \Psi}}
\\
&\;\leq\; \sqrt{\beta + \sqrt{\frac{2}{N}}} \, \sqrt{\frac{N}{N-1} \beta}
\\
&\;\lesssim \beta + \frac{1}{\sqrt{N}}\,,
\end{align*}
where in the fifth step we used Lemma \ref{replacing q's with n's}.

In order to estimate $\mathrm{(III)}^{(p)}$ we introduce a splitting of $w^{(p)}$ into 
``singular'' and ``regular'' parts,
\begin{equation} \label{splitting of w^p}
w^{(p)} \;=\; w^{(p,1)} + w^{(p,2)} \;\deq\; w^{(p)} \, \umat_{\{\abs{w^{(p)}} > a\}} + w^{(p)} 
\, \umat_{\{\abs{w^{(p)}} \leq a\}}\,,
\end{equation}
where $a$ is a positive ($N$-dependent) constant we choose later.
For future reference we record the estimates
\begin{subequations} \label{estimate on w^{(1,2)} for p_0}
\begin{align} \norm{w^{(p,1)}}_{p_0} &\;\leq\; a^{1-p/p_0} \, \norm{w^{(p)}}_p^{p/p_0}\,, 
\label{estimate on w^{(1)} for p_0} \\
\norm{w^{(p,2)}}_2 &\;\leq\; a^{1-p/2} \, \norm{w^{(p)}}_p^{p/2}\,. \label{estimate on w^{(2)} 
for p_0}
\end{align}
\end{subequations}
The proof of \eqref{estimate on w^{(1,2)} for p_0} is elementary; for instance \eqref{estimate 
on w^{(1)} for p_0} follows from
\begin{multline*}
\norm{w^{(p,1)}}_{p_0}^{p_0} \;=\; \int \dd x \; \absb{w^{(p)}}^p \, \absb{w^{(p)}}^{p_0 - p} 
\, \umat_{\{\abs{w^{(p)}} > a\}} \\
\leq\; a^{p_0-p} \int \dd x \; \absb{w^{(p)}}^p \, \umat_{\{\abs{w^{(p)}} > a\}} \;\leq\; 
a^{p_0-p} \int \dd x \; \absb{w^{(p)}}^p\,.
\end{multline*}

Let us start with $\mathrm{(III)}^{(p,1)}$. As in \eqref{w as divergence of xi}, we use the 
representation
\begin{equation*}
w^{(p,1)} \;=\; \nabla \cdot \xi\,.
\end{equation*}
Then \eqref{bound on xi} and \eqref{estimate on w^{(1)} for p_0} imply that
\begin{equation} \label{estimate on xi_2}
\norm{\xi}_2 \;\lesssim\; \norm{w^{(p,1)}}_{p_0} \;\lesssim\; a^{1 - p/p_0}\,.
\end{equation}
Integrating by parts, we find
\begin{align}
2 \absb{\mathrm{(III)}^{(p,1)}} &\;=\; \absb{\scalarb{\Psi}{p_1 p_2 W_{12}^{(p,1)} \, 
\widehat{\nu}\, q_1 q_2\Psi}}
\notag \\
&\;=\; \absb{\scalarb{\Psi}{p_1 p_2 (\nabla_1^\rho X_{12}^\rho) \, \widehat{\nu}\, q_1 
q_2\Psi}}
\notag \\ \label{(III)^{(p,1)} as two terms}
&\;\leq\; \absb{\scalarb{\nabla_1^\rho p_1 p_2 \Psi}{X_{12}^\rho \, \widehat{\nu} \, q_1 q_2 
\Psi}} + \absb{\scalarb{p_1 p_2 \Psi}{X_{12}^\rho \nabla_1^\rho \, \widehat{\nu} \, q_1 q_2 
\Psi}}\,.
\end{align}
Using $\norm{\nabla p} = \norm{\nabla \varphi}$ and Lemma \ref{replacing q's with n's} we find 
that the first term of \eqref{(III)^{(p,1)} as two terms} is bounded by
\begin{align*}
\sqrt{\scalarb{\nabla_1^\rho p_1 \Psi}{p_2 X_{12}^\rho X_{12}^\sigma p_2 \nabla_1^\sigma p_1 
\Psi}} \, \sqrt{\scalarb{\Psi}{\widehat{\nu}^2 q_1 q_2 \Psi}} &\;\lesssim\; \norm{\nabla p} \, 
\norm{\varphi}_\infty \norm{\xi}_2 \, \sqrt{\alpha}
\\
&\;\leq\; \norm{\nabla \varphi} \, \norm{\varphi}_\infty \, a^{1 - p/p_0} \, \sqrt{\beta}
\\
&\;\leq\; \norm{\nabla \varphi} \, \norm{\varphi}_\infty \pb{\beta + a^{2 - 2 p /p_0}}\,,
\end{align*}
where in the second step we used the estimate \eqref{estimate on xi_2}. Next, using Lemma 
\ref{pulling projectors through}, we find that the second term of \eqref{(III)^{(p,1)} as two 
terms} is equal to
\begin{multline*}
\absb{\scalarb{p_1 p_2 \Psi}{X_{12}^\rho (p_1 + q_1) \nabla_1^\rho \, \widehat{\nu} \, q_1 q_2 
\Psi}}
\\
\leq\;
\absb{\scalarb{p_1 p_2 \Psi}{X_{12}^\rho p_1 \, \widehat{\tau_1 \nu} \, \nabla_1^\rho q_1 q_2 
\Psi}} + \absb{\scalarb{p_1 p_2 \Psi}{X_{12}^\rho q_1 \, \widehat{\nu} \, \nabla_1^\rho q_1 q_2 
\Psi}}\,.
\end{multline*}
We estimate the first term (the second is dealt with in exactly the same way):
\begin{align*}
\absb{\scalarb{p_1 p_2 \Psi}{X_{12}^\rho p_1 \, \widehat{\tau_1 \nu} \, \nabla_1^\rho q_1 q_2 
\Psi}} &\;\leq\; \sqrt{\scalarb{\Psi}{p_1 p_2 X_{12}^2 p_1 p_2 \Psi}} \, 
\sqrt{\scalarb{\nabla_1 q_1 \Psi}{\widehat{\tau_1 \nu}^2 \, q_2 \nabla_1 q_1 \Psi}}
\\
&\;\leq\; \sqrt{\norm{p_2 X_{12}^2 p_2}} \sqrt{\frac{1}{N-1} \sum_{i = 2}^N \scalarb{\nabla_1 
q_1 \Psi}{\widehat{n}^{-2} \, q_i \nabla_1 q_1 \Psi}}
\\
&\;\leq\; \norm{\xi}_2 \, \norm{\varphi}_\infty \sqrt{\frac{1}{N-1} \sum_{i = 1}^N 
\scalarb{\nabla_1 q_1 \Psi}{\widehat{n}^{-2} \, q_i \nabla_1 q_1 \Psi}}
\\
&\;\lesssim\; a^{1 - p/p_0} \norm{\varphi}_\infty \sqrt{\frac{N}{N-1}\, \scalarb{\nabla_1 q_1 
\Psi}{\nabla_1 q_1 \Psi}}
\\
&\;\leq\; \norm{\varphi}_\infty \pb{a^{2 - 2 p/p_0} + \norm{\nabla_1 q_1 \Psi}^2}\,.
\end{align*}
Summarizing,
\begin{equation*}
\absb{\mathrm{(III)}^{(p,1)}} \;\lesssim\; \norm{\varphi}_\infty \pB{\beta \norm{\varphi}_{X_1} 
+ \norm{\nabla_1 q_1 \Psi}^2 + a^{2 - 2 p/p_0} \norm{\varphi}_{X_1}}\,.
\end{equation*}

Finally, we estimate \begin{equation}
\mathrm{(III)}^{(p,2)} \;=\; \absb{\scalarb{\Psi}{p_1 p_2 W_{12}^{(p,2)} \, \widehat{\nu}\, q_1 
q_2\Psi}} \;=\; \absb{\scalarb{\Psi}{p_1 p_2 W_{12}^{(p,2)} \, \widehat{\nu} \, 
(\widehat{\chi^{(1)}} + \widehat{\chi^{(2)}}) q_1 q_2\Psi}}\,,
\end{equation}
where
\begin{equation*}
1 \;=\; \chi^{(1)} + \chi^{(2)}\,, \qquad \chi^{(1)}, \chi^{(2)} \;\in\; \{0,1\}^{\{0, 
\dots,N\}}\,,
\end{equation*}
is some partition of the unity to be chosen later.  The need for this partitioning will soon 
become clear.
In order to bound the term with $\chi^{(1)}$, we note that the operator norm of $p_1 p_2 
W_{12}^{(p,2)} q_1 q_2$ on the full space $L^2(\R^{dN})$ is much larger than on its symmetric 
subspace. Thus, as a first step, we symmetrize the operator $p_1 p_2 W_{12}^{(p,2)} q_1 q_2$ in 
coordinate $2$. We get the bound
\begin{align*}
\absb{\scalarb{\Psi}{p_1 p_2 W_{12}^{(p,2)} \, \widehat{\nu} \, \widehat{\chi^{(1)}} \, q_1 
q_2\Psi}} &\;=\; \frac{1}{N-1} \absbb{\scalarbb{\Psi}{\sum_{i = 2}^N p_1 p_i W_{1i}^{(p,2)} \, 
q_i q_1 \, \widehat{\chi^{(1)}}\, \widehat{\nu} \, q_1 \Psi}}
\\
&\;\leq\; \frac{1}{N-1} \, \normb{\widehat{\nu} \, q_1 \Psi} \sqrt{\sum_{i,j = 2}^N 
\scalarb{\Psi}{p_1 p_i W^{(p,2)}_{1i} q_1 q_i \, \widehat{\chi^{(1)}} \, q_1 q_j W^{(p-2)}_{1j} 
p_j p_1 \Psi}}\,.
\end{align*}
Using \begin{equation*}
\normb{\widehat{\nu} \, q_1 \Psi} \;\leq\; \norm{\widehat{n}^{-1} q_1 \Psi} \;\leq\; 1
\end{equation*}
we find
\begin{equation} \label{splitting in A+B for singular potentials}
\absb{\scalarb{\Psi}{p_1 p_2 W_{12}^{(p,2)} \, \widehat{\nu} \, \widehat{\chi^{(1)}} \, q_1 
q_2\Psi}} \;\leq\; \frac{1}{N - 1} \sqrt{A + B}\,,
\end{equation}
where
\begin{align*}
A &\;\deq\; \sum_{2 \leq i \neq j \leq N} \scalarb{\Psi}{p_1 p_i W^{(p,2)}_{1i} q_1 q_i \, 
\widehat{\chi^{(1)}} \, q_j W^{(p,2)}_{1j} p_j p_1 \Psi}\,,
\\
B &\;\deq\; \sum_{i = 2}^N \scalarb{\Psi}{p_1 p_i W^{(p,2)}_{1i} q_1 q_i \, 
\widehat{\chi^{(1)}} \, W^{(p,2)}_{1i} p_i p_1 \Psi}\,.
\end{align*}

The easy part is
\begin{align*}
B &\;\leq\; \sum_{i = 2}^N \scalarb{\Psi}{p_1 p_i \pb{W^{(p,2)}_{1i}}^2 p_i p_1 \Psi} \\
&\;\leq\; \sum_{i = 2}^N \normb{\pb{w^{(p,2)}}^2 * \abs{\varphi}^2}_\infty \scalar{\Psi}{p_1 
p_i \Psi} \\
&\;\leq\; (N-1) \norm{\varphi}_\infty^2 \norm{w^{(p,2)}}_2^2
\\
&\;\lesssim\; N \, a^{2-p} \, \norm{\varphi}_{\infty}^2\,.
\end{align*}
Let us therefore concentrate on
\begin{align*}
A &\;=\; \sum_{2 \leq i \neq j \leq N} \scalarb{\Psi}{p_1 p_i W^{(p,2)}_{1i} q_1 q_i \, 
\widehat{\chi^{(1)}} \, \widehat{\chi^{(1)}}\, q_j W^{(p,2)}_{1j} p_j p_1 \Psi}
\\
&\;=\; \sum_{2 \leq i \neq j \leq N} \scalarb{\Psi}{p_1 p_i q_j \, \widehat{\tau_2 \chi^{(1)}} 
\, W^{(p,2)}_{1i} q_1 W^{(p,2)}_{1j} \, \widehat{\tau_2 \chi^{(1)}} \, q_i p_j p_1 \Psi}
\\
&\;=\; A_1 + A_2\,,
\end{align*}
with $A = A_1 + A_2$ arising from the splitting $q_1 = \umat - p_1$.  We start with
\begin{align*}
\abs{A_1} &\;\leq\; \sum_{2 \leq i \neq j \leq N} \absb{\scalarb{\Psi}{p_1 p_i q_j \, 
\widehat{\tau_2 \chi^{(1)}} \, W^{(p,2)}_{1i} W^{(p,2)}_{1j} \, \widehat{\tau_2 \chi^{(1)}} \, 
q_i p_j p_1 \Psi}}
\\
&\;=\; \sum_{2 \leq i \neq j \leq N} \absb{\scalarb{\Psi}{p_1 p_i q_j \, \widehat{\tau_2 
\chi^{(1)}} \, \sqrt{W^{(p,2)}_{1i}} \sqrt{W^{(p,2)}_{1j}}  \sqrt{W^{(p,2)}_{1i}} 
\sqrt{W^{(p,2)}_{1j}} \, \widehat{\tau_2 \chi^{(1)}} \, q_i p_j p_1 \Psi}}
\\
&\;\leq\; \sum_{2 \leq i \neq j \leq N} \scalarb{\Psi}{\widehat{\tau_2 \chi^{(1)}} \, q_j p_1 
p_i \absb{W_{1i}^{(p,2)}} \absb{W_{1j}^{(p,2)}} p_1 p_i q_j \, \widehat{\tau_2 \chi^{(1)}} \, 
\Psi}\,,
\end{align*}
by Cauchy-Schwarz and symmetry of $\Psi$. Here $\sqrt{\cdot}$ is any complex square root.

In order to estimate this we claim that, for $i \neq j$,
\begin{equation} \label{estimate of two W's}
\normB{p_1 p_i \absb{W^{(p,2)}_{1i}} \absb{W^{(p,2)}_{1j}} p_1 p_i} \;\leq\; 
\normb{\absb{w^{(p,2)}} * \abs{\varphi}^2}_\infty^2\,.
\end{equation}
Indeed, by \eqref{W sandwiched between p's}, we have
\begin{equation*}
p_1 p_i \absb{W^{(p,2)}_{1i}} \absb{W^{(p,2)}_{1j}} p_1 p_i \;=\; p_1 p_i \absb{W^{(p,2)}_{1i}} 
p_i \absb{W^{(p,2)}_{1j}} p_1 \;=\; p_1 p_i \pb{\absb{w^{(p,2)}} * \abs{\varphi}^2}_1 
\absb{W^{(p,2)}_{1j}} p_1\,.
\end{equation*}
The operator $p_1 \pb{\absb{w^{(p,2)}} * \abs{\varphi}^2}_1 \absb{W^{(p,2)}_{1j}} p_1$ is equal 
to $f_j \, p_1$, where
\begin{equation*}
f(x_j) \;=\; \int \dd x_1 \; \ol{\varphi(x_1)} \pb{\absb{w^{(p,2)}} * \abs{\varphi}^2}(x_1) 
\absb{w^{(p,2)}(x_1 - x_j)}  \varphi(x_1)\,.
\end{equation*}
Thus,
\begin{equation*}
\norm{f}_\infty \;\leq\; \normb{\absb{w^{(p,2)}} * \abs{\varphi}^2}_\infty^2\,,
\end{equation*}
from which \eqref{estimate of two W's} follows immediately.

Using \eqref{estimate of two W's}, we get
\begin{align*}
\abs{A_1} &\;\leq\; \sum_{2 \leq i \neq j \leq N} \normb{\absb{w^{(p,2)}} * 
\abs{\varphi}^2}_\infty^2 \normb{\widehat{\tau_2 \chi^{(1)}} q_1 \Psi}^2 \\
&\;\leq\; N^2 \norm{w^{(p)}}_p^2 \, \norm{\varphi}_{L^2 \cap L^\infty}^4 \, 
\scalarb{\Psi}{\widehat{\tau_2 \chi^{(1)}} \, q_1 \Psi}
\\
&\;\lesssim\; N^2 \, \norm{\varphi}_{L^2 \cap L^\infty}^4\, \scalarb{\Psi}{\widehat{\tau_2 
\chi^{(1)}} \, \widehat{n}^2\, \Psi}\,.
\end{align*}
Now let us choose
\begin{equation}
\chi^{(1)}(k) \;\deq\; \umat_{\{k \leq N^{1-\delta}\}}
\end{equation}
for some $\delta \in (0,1)$. Then
\begin{equation*}
(\tau_2 \chi^{(1)}) \, n^2 \;\leq\; N^{-\delta} \end{equation*}
implies
\begin{equation*}
\abs{A_1} \;\lesssim\; \norm{\varphi}_{L^2 \cap L^\infty}^4 \, N^{2 - \delta}\,.
\end{equation*}
Similarly, we find
\begin{align*}
\abs{A_2} &\;\leq\; \sum_{2 \leq i \neq j \leq N} \absb{\scalarb{\Psi}{ q_j \, \widehat{\tau_2 
\chi^{(1)}} \, p_i p_1 W^{(p,2)}_{1i} p_1 W^{(p,2)}_{1j} p_1 p_j \, \widehat{\tau_2 \chi^{(1)}} 
\, q_i \Psi}} \\
&\;\leq\; \sum_{2 \leq i \neq j \leq N} \normb{w^{(p,2)} * \abs{\varphi}^2}_\infty^2 
\scalar{\Psi}{\widehat{\tau_2 \chi^{(1)}} \, q_1 \Psi}
\\
&\;\lesssim\; N^2 \norm{\varphi}_{L^2 \cap L^\infty}^4 N^{-\delta}
\\
&\;=\; \norm{\varphi}_{L^2 \cap L^\infty}^4 N^{2-\delta}\,.
\end{align*}
Thus we have proven
\begin{equation*}
\abs{A} \;\lesssim\; \norm{\varphi}_{L^2 \cap L^\infty}^4 N^{2-\delta}\,.
\end{equation*}
Going back to \eqref{splitting in A+B for singular potentials}, we see that
\begin{equation*}
\absb{\scalarb{\Psi}{p_1 p_2 W_{12}^{(p,2)} \, \widehat{\nu} \, \widehat{\chi^{(1)}} \, q_1 
q_2\Psi}} \;\lesssim\;
\norm{\varphi}_{L^2 \cap L^\infty}^2 N^{-\delta/2} + \norm{\varphi}_{\infty} N^{-1/2}\, a^{1 - 
p/2}\,.
\end{equation*}

What remains is to estimate is the term of $\mathrm{(III)}^{(p,2)}$ containing $\chi^{(2)}$,
\begin{multline*}
\absb{\scalarb{\Psi}{p_1 p_2 W_{12}^{(p,2)} \, \widehat{\nu} \, \widehat{\chi^{(2)}} \, q_1 
q_2\Psi}} \;=\; \frac{1}{N-1} \absbb{\scalarbb{\Psi}{\sum_{i = 2}^N p_1 p_i W_{1i}^{(p,2)} \, 
q_i q_1 \, \widehat{\chi^{(2)}}\, \widehat{\nu}^{1/2} \, \widehat{\nu}^{1/2} \, q_1 \Psi}}
\\
\leq\; \frac{1}{N-1} \, \normb{\widehat{\nu}^{1/2} \, q_1 \Psi} \sqrt{\sum_{i,j = 2}^N 
\scalarb{\Psi}{p_1 p_i W^{(p,2)}_{1i} q_1 q_i \, \widehat{\chi^{(2)}} \, \widehat{\nu} \, q_1 
q_j W^{(p-2)}_{1j} p_j p_1 \Psi}}\,.
\end{multline*}
Using \begin{equation*}
\normb{\widehat{\nu}^{1/2} \, q_1 \Psi} \;\leq\; \sqrt{\scalar{\Psi}{\widehat{n}^{-1} \, 
\widehat{n}^2 \, \Psi}} \;=\; \sqrt{\beta}
\end{equation*}
we find
\begin{equation} \label{splitting in A+B for singular potentials for chi_2}
\absb{\scalarb{\Psi}{p_1 p_2 W_{12}^{(p,2)} \, \widehat{\nu} \, \widehat{\chi^{(2)}} \, q_1 
q_2\Psi}} \;\leq\; \frac{\sqrt{\beta}}{N - 1} \sqrt{A + B}\,,
\end{equation}
where
\begin{align*}
A &\;\deq\; \sum_{2 \leq i \neq j \leq N} \scalarb{\Psi}{p_1 p_i W^{(p,2)}_{1i} q_1 q_i \, 
\widehat{\chi^{(2)}} \, \widehat{\nu} \, q_j W^{(p,2)}_{1j} p_j p_1 \Psi}\,,
\\
B &\;\deq\; \sum_{i = 2}^N \scalarb{\Psi}{p_1 p_i W^{(p,2)}_{1i} q_1 q_i \, 
\widehat{\chi^{(2)}} \, \widehat{\nu} \, W^{(p,2)}_{1i} p_i p_1 \Psi}\,.
\end{align*}
Since
\begin{equation*}
\chi^{(2)}(k) \;=\; \umat_{\{k > N^{1-\delta}\}}
\end{equation*}
we find
\begin{equation*}
\chi^{(2)} \, \nu \;\leq\; \chi^{(2)} \, n^{-1}\;\leq\; N^{\delta/2}\,.
\end{equation*}
Thus, $\norm{q_1 q_i \, \widehat{\chi^{(2)}} \, \widehat{\nu}} \;\leq\; N^{\delta/2}$ and we 
get
\begin{multline*}
B \;\leq\; N^{\delta/2}\sum_{i = 2}^N \scalarb{\Psi}{p_1 p_i \pb{W^{(p,2)}_{1i}}^2 p_i p_1 
\Psi} \;\leq\; N^{1 + \delta/2} \, \normb{\pb{w^{(p,2)}}^2 * \abs{\varphi}^2}_\infty
\\
\;\leq\; N^{1 + \delta/2} \, \norm{w^{(p,2)}}_2^2 \, \norm{\varphi}_\infty^2 \;\lesssim\; N^{1 
+ \delta/2} \, a^{2 - p}\, \norm{\varphi}_\infty^2\,,
\end{multline*}
by \eqref{estimate on w^{(2)} for p_0}.

Next, using Lemma \ref{pulling projectors through}, we find
\begin{align*}
A &\;=\; \sum_{2 \leq i \neq j \leq N} \scalarb{\Psi}{p_1 p_i q_j W^{(p,2)}_{1i} \, 
\widehat{\chi^{(2)}} \, \widehat{\nu}^{1/2} \, q_1 \, \widehat{\chi^{(2)}} \, 
\widehat{\nu}^{1/2} \, W^{(p,2)}_{1j} q_i p_j p_1 \Psi}
\\
&\;=\; \sum_{2 \leq i \neq j \leq N} \scalarb{\Psi}{p_1 p_i q_j \, \widehat{\tau_2 \chi^{(2)}} 
\, \widehat{\tau_2 \nu}^{1/2} \, W^{(p,2)}_{1i} q_1 W^{(p,2)}_{1j} \, \widehat{\tau_2 
\chi^{(2)}} \, \widehat{\tau_2 \nu}^{1/2} \, q_i p_j p_1 \Psi}
\\
&\;=\; A_1 + A_2\,,
\end{align*}
where, as above, the splitting $A = A_1 + A_2$ arises from writing $q_1 = \umat - p_1$.
Thus,
\begin{align*}
\abs{A_1} &\;\leq\; \sum_{2 \leq i \neq j \leq N} \absb{\scalarb{\Psi}{p_1 p_i q_j \, 
\widehat{\tau_2 \chi^{(2)}} \, \widehat{\tau_2 \nu}^{1/2} \, W^{(p,2)}_{1i} W^{(p,2)}_{1j} \, 
\widehat{\tau_2 \chi^{(2)}} \, \widehat{\tau_2 \nu}^{1/2} \, q_i p_j p_1 \Psi}}
\\
&\;=\; \sum_{2 \leq i \neq j \leq N} \absb{\scalarb{\Psi}{p_1 p_i q_j \, \widehat{\tau_2 
\chi^{(2)}} \, \widehat{\tau_2 \nu}^{1/2} \, \sqrt{W^{(p,2)}_{1i}} \sqrt{W^{(p,2)}_{1j}} 
\sqrt{W^{(p,2)}_{1i}} \sqrt{W^{(p,2)}_{1j}} \, \widehat{\tau_2 \chi^{(2)}} \, \widehat{\tau_2 
\nu}^{1/2} \, q_i p_j p_1 \Psi}}
\\
&\;\leq\; \sum_{2 \leq i \neq j \leq N} \scalarb{\Psi}{q_j \, \widehat{\tau_2 \chi^{(2)}} \, 
\widehat{\tau_2 \nu}^{1/2} \, p_1 p_i \absb{W^{(p,2)}_{1i}} \absb{W^{(p,2)}_{1j}} p_i p_1 \, 
\widehat{\tau_2 \chi^{(2)}} \, \widehat{\tau_2 \nu}^{1/2} \, q_j  \Psi}\,,
\end{align*}
by Cauchy-Schwarz and symmetry of $\Psi$. Using \eqref{estimate of two W's} we get
\begin{align*}
\abs{A_1} &\;\leq\; N^2 \, \normb{\absb{w^{(p,2)}} * \abs{\varphi}^2}_\infty^2 \, 
\scalarb{\Psi}{\widehat{\tau_2 \nu} \, q_1 \Psi}
\\
&\;\leq\; N^2 \norm{w^{(p,2)}}_p^2 \, \norm{\varphi}_{L^2 \cap L^\infty}^4 \, 
\scalar{\Psi}{\widehat{n} \, \Psi}
\\
&\;\lesssim\; N^2 \, \norm{\varphi}_{L^2 \cap L^\infty}^4\beta\,.
\end{align*}
Similarly,
\begin{align*}
\abs{A_2} &\;\leq\; \sum_{2 \leq i \neq j \leq N} \absb{\scalarb{\Psi}{p_i q_j \, 
\widehat{\tau_2 \chi^{(2)}} \, \widehat{\tau_2 \nu}^{1/2} \,p_1 W^{(p,2)}_{1i} p_1 
W^{(p,2)}_{1j} p_1 \, \widehat{\tau_2 \chi^{(2)}} \, \widehat{\tau_2 \nu}^{1/2} \, q_i p_j 
\Psi}}
\\
&\;\leq\; \sum_{2 \leq i \neq j \leq N} \normb{w^{(p,2)} * \abs{\varphi}^2}_\infty^2 
\scalarb{\Psi}{\widehat{\tau_2 \nu} \, q_1 \Psi}
\\
&\;\leq\; N^2 \norm{w^{(p)}}_p^2 \, \norm{\varphi}_{L^2 \cap L^\infty}^4 \, 
\scalar{\Psi}{\widehat{n} \, \Psi}
\\
&\;\lesssim\; N^2 \, \norm{\varphi}_{L^2 \cap L^\infty}^4 \, \beta\,.
\end{align*}
Plugging all this back into \eqref{splitting in A+B for singular potentials for chi_2}, we find 
that
\begin{equation*}
\absb{\scalarb{\Psi}{p_1 p_2 W_{12}^{(p,2)} \, \widehat{\nu} \, \widehat{\chi^{(2)}} \, q_1 
q_2\Psi}} \;\lesssim\;
\beta \pb{\norm{\varphi}_{L^2 \cap L^\infty}^2 + \norm{\varphi}_\infty} + \norm{\varphi}_\infty 
a^{2 - p} N^{\delta/2 - 1}
\end{equation*}

Summarizing:
\begin{equation*}
\absb{\mathrm{(III)}^{(p,2)}} \;\lesssim\; \pb{1 + \norm{\varphi}_{L^2 \cap L^\infty}^2} 
\pB{\beta + a^{2 - p} \, N^{\delta/2 - 1} + N^{-\delta / 2} + N^{-1/2} a^{1 - p/2}}\,,
\end{equation*}
from which we deduce
\begin{multline*}
\absb{\mathrm{(III)}^{(p)}} \;\lesssim\; \norm{\varphi}_\infty \norm{\nabla_1 q_1 \Psi}^2 \\
{}+{} \pb{1 + \norm{\varphi}_{X_1 \cap L^\infty}} \pB{\beta + a^{2 - p} \, N^{\delta/2 - 1} + 
N^{-\delta / 2} + N^{-1/2} a^{1 - p/2} + a^{2 - 2 p / p_0}}\,.
\end{multline*}
Let us set $a \equiv a_N = N^\zeta$ and optimize in $\delta$ and $\zeta$. This yields the 
relations
\begin{equation*}
\zeta (2 - p) + \delta \;=\; 1 \,,\qquad -\frac{\delta}{2} \;=\; 2 \zeta \pbb{1 - 
\frac{p}{p_0}}\,,
\end{equation*}
which imply
\begin{equation*}
\frac{\delta}{2} \;=\; \frac{p/p_0 - 1}{2 p/p_0 - p/2 - 1}\,,
\end{equation*}
with $\delta \leq 1$.
Thus,
\begin{equation*}
\absb{\mathrm{(III)}^{(p)}} \;\lesssim\; \norm{\varphi}_\infty \norm{\nabla_1 q_1 \Psi}^2 + 
\pb{1 + \norm{\varphi}_{X_1 \cap L^\infty}} \pB{\beta + N^{-\eta}}\,,
\end{equation*}
where $\eta = \delta/2$ satisfies \eqref{definition of eta}.

\step{Conclusion of the proof.} We have shown that
\begin{equation*}
\dot{\beta} \;\lesssim\; \norm{\varphi}_{L^2 \cap L^\infty} \norm{\nabla_1 q_1 \Psi}^2 + \pb{1 
+ \norm{\varphi}_{X_1 \cap L^\infty}} \pb{\beta + N^{-\eta}}\,.
\end{equation*}
Using Lemma \ref{lemma: energy estimate} we find
\begin{equation} \label{final bound on beta}
\dot{\beta} \;\lesssim\; \pB{1 + \norm{\varphi}^3_{X_1^2 \cap L^\infty}} \pbb{\beta + E^\Psi - 
E^\varphi + \frac{1}{N^\eta}}\,.
\end{equation}
The claim then follows from the Gr\"onwall estimate \eqref{Gronwall}.

\subsection{A remark on time-dependent external potentials} \label{section: remark on 
time-dependent potentials}
Theorem \ref{theorem for singular potentials} can be extended to time-dependent external 
potentials $h(t)$ without too much sweat. The only complication is that energy is no longer 
conserved. We overcome this problem by observing that, while the energies $E^\Psi(t)$ and 
$E^\varphi(t)$ exhibit large variations in $t$, their difference remains small. In the 
following we estimate the quantity $E^\Psi(t) - E^\varphi(t)$ by controlling its time 
derivative.

We need the following assumptions, which replace Assumptions (B1) -- (B3).

\begin{itemize}
\item[(B1')]
The Hamiltonian $h(t)$ is self-adjoint and bounded from below. We assume that there is an 
operator $h_0 \geq 0$ that  such that $0 \leq h(t) \leq h_0$ for all $t$. We define the Hilbert 
space $X_N \;=\; \mathcal{Q}\pb{\sum_i (h_0)_i}$ as in (A1), and the space $X_1^2 = \cal 
Q(h_0^2)$ as in (B5) using $h_0$. We also assume that there are time-independent constants 
$\kappa_1, \kappa_2 > 0$ such that
\begin{equation*}
-\Delta \;\leq\; \kappa_1\, h(t) + \kappa_2
\end{equation*}
for all $t$.

We make the following assumptions on the differentiability of $h(t)$. The map $t \mapsto 
\scalar{\psi}{h(t) \psi}$ is continuously differentiable for all $\psi \in X_1$, with 
derivative $\scalar{\psi}{\dot{h}(t) \psi}$ for some self-adjoint operator $\dot{h}(t)$. 
Moreover, we assume that the quantities
\begin{equation*}
\scalar{\varphi(t)}{\dot{h}(t)^2 \varphi(t)}\,, \qquad \normb{(\umat + h(t))^{-1/2} \, 
\dot{h}(t) \, (\umat + h(t))^{-1/2}}
\end{equation*}
are continuous and finite for all $t$.
\item[(B2')]
The Hamiltonian $H_N(t)$ is self-adjoint and bounded from below.
We assume that $\mathcal{Q}(H_N(t)) \subset X_N$ for all $t$.
We also assume that the $N$-body propagator $U_N(t,s)$, defined by
\begin{equation*}
\ii \partial_t U_N(t,s) = H_N(t) U_N(t,s) \,, \qquad U_N(s,s) = \umat\,,
\end{equation*}
exists and satisfies $U_N(t,0)\Psi_{N,0} \in \mathcal{Q}(H_N(t))$ for all $t$.
\item[(B3')]
There is a time-independent constant $\kappa_3 \in (0,1)$ such that
\begin{equation*}
0 \;\leq\; (1 - \kappa_3) (h_1(t) + h_2(t)) + W_{12}
\end{equation*}
for all $t$.
\end{itemize}

\begin{theorem}
Assume that Assumptions (B1') -- (B3'), (B4), and (B5) hold. Then there is a continuous 
nonnegative function $\phi$, independent of $N$ and $\Psi_{N,0}$, such that
\begin{equation*}
\beta_N(t) \;\leq\; \phi(t) \pbb{\beta_N(0) + E^\Psi_N(0) - E^\varphi(0) + \frac{1}{N^\eta}}\,,
\end{equation*}
with $\eta$ defined in \eqref{definition of eta}.
\end{theorem}

\begin{proof}
We start by deriving an upper bound on the energy difference $\cal E(t) \deq E^\Psi(t) - 
E^\varphi(t)$. Assumptions (B1') and (B2') and the fundamental theorem of calculus imply
\begin{equation*}
\cal E(t) \;=\; \cal E(0) + \int_0^t \dd s \, \pB{\underbrace{\scalar{\Psi(s)}{\dot{h}_1(s) 
\Psi(s)} - \scalar{\varphi(s)}{\dot{h}(s) \varphi(s)}}_{\eqd \, G(s)}}\,.
\end{equation*}
By inserting $\umat = p_1(s) + q_1(s)$ on both sides of $\dot{h}_1(s)$ we get (omitting the 
time argument $s$)
\begin{equation} \label{derivative of the energy expanded}
G \;=\; \scalar{\Psi}{p_1 \dot{h}_1 p_1 \Psi} - \scalar{\varphi}{\dot{h} \varphi} + 2 \re 
\scalar{\Psi}{p_1 \dot{h}_1 q_1 \Psi} + \scalar{\Psi}{q_1 \dot{h}_1 q_1 \Psi}\,.
\end{equation}
The first two terms of \eqref{derivative of the energy expanded} are equal to
\begin{equation*}
\pb{\scalar{\Psi}{p_1 \Psi} - 1}\scalar{\varphi}{\dot{h} \varphi} \;=\; \alpha 
\scalar{\varphi}{\dot{h} \varphi} \;\leq\; \beta \abs{\scalar{\varphi}{\dot{h} \varphi}} \,.  
\end{equation*}
The third term of \eqref{derivative of the energy expanded} is bounded, using Lemmas 
\ref{replacing q's with n's} and \ref{pulling projectors through}, by
\begin{align*}
2 \absb{\scalarb{\Psi}{p_1 \dot{h}_1 \, \widehat{n}^{1/2}\, \widehat{n}^{-1/2}\, q_1 \Psi}} 
&\;=\;
2 \absb{\scalarb{\dot{h}_1 p_1 \, \widehat{\tau_1 n}^{1/2}\, \Psi}{\widehat{n}^{-1/2}\, q_1 
\Psi}}
\\
&\;\leq\; \sqrt{\scalarb{\widehat{\tau_1 n}^{1/2} \, \Psi}{p_1 \dot{h}_1^2 p_1 \, 
\widehat{\tau_1 n}^{1/2} \, \Psi}}\, \normb{\widehat{n}^{-1/2} \, q_1 \Psi}
\\
&\;\leq\; \sqrt{\abs{\scalar{\varphi}{\dot{h}^2 \varphi}}} \, 
\sqrt{\scalarb{\Psi}{\widehat{\tau_1 n} \, \Psi}} \, \sqrt{\scalarb{\Psi}{ \widehat{n}^{-1} \, 
q_1 \Psi}}
\\
&\;\leq\; \sqrt{\abs{\scalar{\varphi}{\dot{h}^2 \varphi}}}\sqrt{\beta + \frac{1}{\sqrt{N}}} 
\sqrt{\beta}\,,
\\
&\;\lesssim\; \sqrt{\abs{\scalar{\varphi}{\dot{h}^2 \varphi}}} \pbb{\beta + 
\frac{1}{\sqrt{N}}}\,.
\end{align*}
The last term of \eqref{derivative of the energy expanded} is equal to
\begin{multline*}
\scalarb{\Psi}{q_1 (\umat + h_1)^{1/2} (\umat + h)^{-1/2} \dot{h}_1 (\umat + h_1)^{-1/2} (\umat 
+ h)^{1/2} q_1 \Psi} \\
\leq\; \normb{(\umat + h)^{-1/2} \dot{h} (\umat + h)^{-1/2}} \, \normb{(\umat + h_1)^{1/2} q_1 
\Psi}^2\,.
\end{multline*}
Thus, using Assumption (B1') we conclude that
\begin{equation} \label{first estimate for G}
G(t) \;\leq\; C(t) \pbb{\beta(t) + \frac{1}{\sqrt{N}} + \normb{h_1(t)^{1/2} q_1(t) \Psi(t)}^2}
\end{equation}
for all $t$. Here, and in the following, $C(t)$ denotes some continuous nonnegative function 
that does not depend on $N$. 

Next, we observe that, under Assumptions (B1') -- (B3'), the proof of Lemma \ref{lemma: energy 
estimate} remains valid for time-dependent one-particle Hamiltonians. Thus, \eqref{energy 
estimate for h} implies
\begin{equation*}
\normb{h_1(t)^{1/2} q_1(t) \Psi(t)}^2 \;\lesssim\; \cal E(t) + \pb{1 + \norm{\varphi(t)}_{X_1^2 
\cap L^\infty}^2} \pbb{\beta(t) + \frac{1}{\sqrt{N}}}\,.
\end{equation*}
Plugging this into \eqref{first estimate for G} yields
\begin{equation*}
G(t) \;\leq\; C(t) \pbb{\beta(t) + \frac{1}{\sqrt{N}} + \cal E(t)}\,.
\end{equation*}
Therefore,
\begin{equation} \label{main estimate on energy difference}
\cal E(t) \;\leq\; \cal E(0) + \int_0^t \dd s \, C(s) \pbb{\beta(s) + \cal E(s) + 
\frac{1}{\sqrt{N}}}\,,
\end{equation}

Next, we observe that, under Assumptions (B1') -- (B3'), the derivation of the estimate 
\eqref{final bound on beta} in the proof of Theorem \ref{theorem for singular potentials} 
remains valid for time-dependent one-particle Hamiltonians. Therefore,
\begin{equation} \label{main estimate on beta for time-dependent case}
\beta(t) \;\leq\; \beta(0) + \int_0^t \dd s \, C(s) \pbb{\beta(s) + \cal E(s) + 
\frac{1}{N^\eta}}\,.
\end{equation}
Applying Gr\"onwall's lemma to the sum of \eqref{main estimate on energy difference} and 
\eqref{main estimate on beta for time-dependent case} yields
\begin{equation*}
\beta(t) + \cal E(t) \;\leq\; \pb{\beta(0) + \cal E(0)} \, \ee^{\int_0^t C} + \frac{1}{N^\eta} 
\int_0^t \dd s \, C(s) \, \ee^{\int_0^t C}\,.
\end{equation*}
Plugging this back into \eqref{main estimate on beta for time-dependent case} yields
\begin{equation*}
\beta(t) \;\leq\; C(t) \pbb{\beta(0) + \cal E(0) + \frac{1}{N^\eta}}\,,
\end{equation*}
which is the claim.
\end{proof}

\providecommand{\bysame}{\leavevmode\hbox to3em{\hrulefill}\thinspace}
\providecommand{\MR}{\relax\ifhmode\unskip\space\fi MR }
\providecommand{\MRhref}[2]{%
  \href{http://www.ams.org/mathscinet-getitem?mr=#1}{#2}
}
\providecommand{\href}[2]{#2}

\end{document}